\def\year{2022}\relax
\documentclass[letterpaper]{article} 
\usepackage{aaai22}  
\usepackage{times}  
\usepackage{helvet}  
\usepackage{courier}  
\usepackage[hyphens]{url}  
\usepackage{graphicx} 
\usepackage{natbib}  
\usepackage{caption} 
\DeclareCaptionStyle{ruled}{labelfont=normalfont,labelsep=colon,strut=off} 
\usepackage{algorithm}
\usepackage{algorithmic}

\usepackage{newfloat}
\usepackage{listings}
\floatstyle{ruled}
\newfloat{listing}{tb}{lst}{}
\floatname{listing}{Listing}

\usepackage{amssymb}
\usepackage{amsthm}
\usepackage{mathtools}
\usepackage{color,xcolor}

 \usepackage{hyperref} 
\hypersetup{colorlinks=true, linkcolor=blue, citecolor=blue, anchorcolor=blue, urlcolor=blue}  

\newtheorem{assumption}{Assumption}

\newtheorem{lemma}{Lemma}
\newtheorem{theorem}{Theorem}

\newcommand{\E}{\mathbf{E}} 
\newcommand{\opt}{\mathrm{opt}} 
\DeclareMathOperator*{\argmax}{argmax}
\DeclareMathOperator*{\tr}{tr} 

\newcommand{\compilefullversion}{true}
\ifthenelse{\equal{\compilefullversion}{false}}{%
	\newcommand{\OnlyInFull}[1]{}
	\newcommand{\OnlyInShort}[1]{#1}
}{%
	\newcommand{\OnlyInFull}[1]{#1}%
	\newcommand{\OnlyInShort}[1]{}%
}%

\newcommand{\compilehidecomments}{true}

\ifthenelse{ \equal{\compilehidecomments}{true} }{%
	\newcommand{\wei}[1]{}
	\newcommand{\zhijie}[1]{}
	\newcommand{\jialin}[1]{}
	\newcommand{\xiaoming}[1]{}
}{
	\newcommand{\wei}[1]{{\color{blue!50!black}  [\text{Wei:} #1]}}
	\newcommand{\zhijie}[1]{{\color{red!60!black} [\text{Zhijie:} #1]}}
	\newcommand{\jialin}[1]{{\color{brown!60!black} [\text{Jialin:} #1]}}
	\newcommand{\xiaoming}[1]{{\color{green!60!black} [\text{Xiaoming:} #1]}}
}

\title{Online Influence Maximization under the Independent Cascade Model \\with Node-Level Feedback\thanks{The published version of this paper claims there are OIM algorithms with node-level feedback that use standard offline oracles.
However, we later found there are severe bugs in our initial algorithms.
In this latest version, we retrieve most of our results by presenting an OIM algoithm that uses pair oracles.}}
\author{
    Zhijie Zhang,\textsuperscript{\rm 1,2}
    Wei Chen,\textsuperscript{\rm 3}
    Xiaoming Sun,\textsuperscript{\rm 1,2}
    Jialin Zhang\textsuperscript{\rm 1,2,}\thanks{Corresponding author.}
}
\affiliations{
    \textsuperscript{\rm 1} Institute of Computing Technology, Chinese Academy of Sciences, Beijing, China\\
    \textsuperscript{\rm 2} University of Chinese Academy of Sciences, Beijing, China\\
    \textsuperscript{\rm 3} Microsoft Research Asia, Beijing, China

    \{zhangzhijie,sunxiaoming,zhangjialin\}@ict.ac.cn, weic@microsoft.com
}

\usepackage{bibentry}

\begin{document}

\maketitle

\begin{abstract}
We study the online influence maximization (OIM) problem in social networks, where the learner repeatedly chooses seed nodes to generate cascades, observes the cascade feedback, and gradually learns the best seeds that generate the largest cascade in multiple rounds. In the demand of the real world, we work with \emph{node-level feedback} instead of the common \emph{edge-level feedback} in the literature. The edge-level feedback reveals all edges that pass through information in a cascade, whereas the node-level feedback only reveals the activated nodes with timestamps. The node-level feedback is arguably more realistic since in practice it is relatively easy to observe who is influenced but very difficult to observe from which relationship (edge) the influence comes.
Previously, there is a nearly optimal $\tilde{O}(\sqrt{T})$-regret algorithm for OIM problem under the \emph{linear threshold} (LT) diffusion model with node-level feedback.
It remains unknown whether the same algorithm exists for the \emph{independent cascade} (IC) diffusion model.
In this paper, we resolve this open problem by presenting an $\tilde{O}(\sqrt{T})$-regret algorithm for OIM problem under the IC model with node-level feedback.
\end{abstract}

\section{Introduction}

Social networks have gained great attention in the past decades as a model for describing relationships between humans.
Typically, researchers show great interest in how information, ideas, news, influence, etc spread over social networks, starting from a small set of nodes called \emph{seeds}.
To this end, a variety of diffusion models are proposed to formulate the propagation in reality, and the most well-known ones are the \emph{independent cascade} (IC) model and the \emph{linear threshold} (LT) model \cite{KempeKT03}.
A corresponding optimization problem, known as \emph{influence maximization} (IM), asks how to maximize the influence spread, under a specific diffusion model, by selecting a limited number of ``good'' seeds.
The problem has found enormous applications, including advertising, viral marketing, news transmission, etc.

In the canonical setting, the IM problem takes as input a social network, which is formulated as an edge-weighted directed graph.
The problem is NP-hard but can be well-approximated \cite{KempeKT03}.
For the past decade, more efficient and effective algorithms have been designed \cite{BorgsBCL14,TangXS14,TangSX15}, leading to an almost complete resolution of the problem.
However, the canonical IM is sometimes difficult to apply in practice, as edge parameters of the network are often \emph{unknown} in many scenarios.
A possible way to circumvent such difficulty is to learn the edge parameters from past observed diffusion cascades, and then maximize the influence based on the learned parameters.
The learning task is referred to as \emph{network inference}, and has been extensively studied in the literature \cite{Gomez-RodriguezLK10,MyersL10,Gomez-RodriguezBS11,DuSSY12,NetrapalliS12,AbrahaoCKP13,DaneshmandGSS14,DuSGZ13,DuLBS14,NarasimhanPS15,Pouget-AbadieH15,HeX0L16,ChenSZZ21}.
However, this approach does not take into account the cost of the learning process and fails to balance between exploration and exploitation when future diffusion cascades come.
This motivates the study of \emph{online influence maximization (OIM)} problem considered in this paper.

\begin{table*}[t]
	\centering
	\begin{tabular}{ccccc}
		\hline
		Feedback & Diffusion model & Regret & Pair oracle & Reference \\
		\hline
		Edge-level  & IC & $\widetilde{O}(n^3\sqrt{T})$  & No & \citet{WangC17} \\
		Node-level  & LT & $\widetilde{O}(n^{9/2}\sqrt{T})$ & Yes & \citet{LiKTLC20} \\
		Node-level  & IC & $\widetilde{O}(n^{7/2}\sqrt{T}/\gamma)$ & Yes & Theorem \ref{thm: regret2-IC} \\
		\hline
	\end{tabular}
    \caption{Comparison of results on OIM problems}
    \label{tab: results}
\end{table*}

In OIM, the learner faces an unknown social network and runs $T$ rounds in total.
At each round, the learner chooses a seed set to generate cascades, observes the cascade feedback, and receives the influence value as a reward.
The goal is to maximize the influence values received over $T$ rounds, or equivalently, to minimize the cumulative regret compared with the optimal seed set that generates the largest influence.
The most widely studied feedback in the literature is the \emph{edge-level} feedback \cite{ChenWY13,ChenWYW16,WangC17,WenKVV17,WuLWCW19}, where the learner can observe whether an edge passes through the information received by its start point.
The \emph{node-level} feedback was only investigated very recently in \cite{VaswaniLS15,LiKTLC20}, where the learner can only observe which nodes receive the information at each time step during a diffusion process. In practice, the node-level feedback is more realistic than the edge-level feedback, not only because it reveals less information, but also because it is usually easy to observe who is influenced but very difficult to observe from which edge the influence comes from. For example, in the social network platform, it is easy to learn whether and when a user buys some specific product or service but is difficult to learn based on whose recommendations or comments the user makes such a decision.

In light of this, it is interesting to study the OIM problem with node-level feedback.
For the LT model, \citet{LiKTLC20} recently presents a nearly optimal $\widetilde{O}(\mathrm{poly}(|G|)\sqrt{T})$-regret algorithm, at the cost of invoking the so-called offline pair oracles instead of standard oracles.
For the IC model, it remains unknown whether the same regret bound can be achieved and this has been an interesting open question in the field.


\paragraph{Our contribution.}
In this paper, we resolve the aforementioned open question and present the first $\widetilde{O}(\mathrm{poly}(|G|)\sqrt{T})$-regret algorithm for OIM problem under the IC model with node-level feedback.
Our algorithm also needs to invoke pair oracles since node-level feedback reveals less information.
We compare our result with previous ones in Table \ref{tab: results}.

In the technical part, our main contribution is a novel adaptation of the maximum likelihood estimation (MLE) approach which can learn the network parameters and their confidence ellipsoids based on the node-level feedback.
We believe this technique is of independent interest and may inspire other results in the field.
Besides, we prove the GOM bounded smoothness for the IC model, which is crucial for learning influence functions.
The same property is also shown for the LT model in \cite{LiKTLC20}.

\paragraph{Related work.}
The (offline) influence maximization problem has received great attentions in the past two decades.
We refer interested readers to the surveys of \cite{ChenLC13,LiFWT18} for an overall understanding.

The online influence maximization problem falls into the field of multi-armed bandits (MAB), a prosperous research area that dates back to 1933 \cite{Thompson33}.
In the classical multi-armed \emph{stochastic bandits} \cite{Robbins52,LaiR85}, there is a set of $n$ arms, each of which is associated with a reward specified by some \emph{unknown} distribution.
At each round $t$, the learner chooses an arm and receives a reward sampled from the corresponding distribution.
The goal is to maximize the total expected rewards received over $T$ rounds.
The model was later generalized to the multi-armed \emph{stochastic linear bandits} \cite{AuerCF02}, where each arm is associated with a characteristic vector and its reward is given by the inner product of the vector and an unknown parameter vector.
This model was extensively studied in the literature \cite{DaniHK08,LiCLS10,RusmevichientongT10,Abbasi-YadkoriPS11}.
Further generalizations include \emph{combinatorial} multi-armed bandits (CMAB)  and CMAB with probabilistically triggered arms (CMAB-T) \cite{ChenWY13,ChenWYW16,WangC17}, where a subset of arms, called the \emph{super-arm}, can be chosen, and the reward is defined over super-arms and may be non-linear.
Besides, the arms beyond the chosen super-arm may also be triggered and observed.
CMAB-T is a quite general bandits framework and indeed contains OIM with edge-level feedback as a special case.
However, OIM with node-level feedback does not fit into the CMAB-T framework.

OIM has been studied extensively in the literature.
For edge-level feedback, existing work \cite{ChenWY13,LeiMMCS15,ChenWYW16,WangC17,WenKVV17,WuLWCW19} present both theoretical and heuristic results.
The node-level feedback was first proposed in \cite{VaswaniLS15}.
However, only heuristic algorithms were presented.
Very recently, an $\widetilde{O}(\sqrt{T})$-regret algorithm was presented for the LT model with node-level feedback using pair-oracles in \cite{LiKTLC20}.
However, it remains unknown whether the same result holds for the IC model.

\section{Preliminaries}

\subsection{Notations}

Given a vector $x\in\mathbb{R}^d$, its transpose is denoted by $x^\top$.
The Euclidean norm of $x$ is denoted by $\|x\|$.
For a positive definite matrix $M\in\mathbb{R}^{d\times d}$, the weighted Euclidean norm of $x$ is defined as $\|x\|_M=\sqrt{x^{\top}Mx}$.
The minimum eigenvalue of $M$ is denoted by $\lambda_{\min}(M)$, and its determinant and trace are denoted by $\det[M]$ and $\tr[M]$, respectively.
For a real-valued function $\mu:\mathbb{R}\rightarrow\mathbb{R}$, its first and second derivatives are denoted by $\dot{\mu}$ and $\ddot{\mu}$, respectively.


\subsection{Social Network}

A social network is a weighted directed graph $G=(V,E)$ with a node set $V$ of $n=|V|$ nodes and an edge set $E$ of $m=|E|$ edges.
Each edge $e\in E$ is associated with a weight or probability $p(e)\in[0,1]$.
The edge probability vector is then denoted by $p=(p(e))_{e\in E}$, which describes the graph completely.
For a node $v\in V$, let $N(v)=N^{in}(v)$ be the set of in-neighbors of $v$ and $d_v=|N(v)|$ be its in-degree.
The maximum in-degree of the graph is denoted by $D=\max_{v\in V} d_v$.
In this paper, we use $E_v$ to denote the set of incoming edges of $v$ and $p_v=(p(e))_{e\in E_v}\in [0,1]^{d_v}$ to denote the probability vector corresponding to these edges.
The $e$-th entry of $p_v$ is denoted by $p_v(e)$.
Thus, $p(e)$ and $p_v(e)$ refers to the same edge probability and we will use them interchangeably throughout the paper.
For an edge $e=(u,v)\in E_v$, we use $e_{uv}$ to explicitly indicate $e$'s endpoints.
Let $\chi(e_{uv})\in\{0,1\}^{d_v}$ be the characteristic vector of $e_{uv}$ over $E_v$ such that all entries of $\chi(e_{uv})$ are $0$ except that the entry
corresponding to $e_{uv}$ is $1$.
The characteristic vector of a subset $E'\subseteq E_v$ is then defined as $\chi(E')\coloneqq \sum_{e\in E'}\chi(e)\in\{0,1\}^{d_v}$.
For simplicity, we define $x_e\coloneqq \chi(e)$.

\subsection{Offline Influence Maximization}

The input of the offline problem is a social network, over which the information spreads.
A node $v\in V$ is called \emph{active} if it receives the information and \emph{inactive} otherwise.
We first describe the independent cascade (IC) diffusion model.

In the IC model, the diffusion proceeds in discrete time steps $\tau=0,1,2,\cdots$.
At the beginning of the diffusion ($\tau=0$), there is an initially active set $S_0$ of nodes called \emph{seeds}.
For $\tau\geq 1$, the active node set $S_{\tau}$ after time $\tau$ is generated as follows.
First, let $S_{\tau}=S_{\tau-1}$.
Next, for each $v\in V\setminus S_{\tau-1}$, every node $u\in N(v)\cap (S_{\tau-1}\setminus S_{\tau-2})$ will try to \emph{activate} $v$ \emph{independently} with probability $p(e_{uv})$ (let $S_{-1}=\emptyset$).
Hence, $v$ will be activated with probability $1-\prod_{u\in N(v)\cap (S_{\tau-1}\setminus S_{\tau-2})}(1-p(e_{uv}))$ and be added into $S_{\tau}$ once being activated.
The diffusion terminates if $S_{\tau}=S_{\tau-1}$ for some $\tau$ and therefore it proceeds in at most $n$ time steps.
Let $(S_0,S_1,\cdots,S_{n-1})$ be the sequence of the active node sets during the diffusion process, where $S_{\tau}$ denotes the active node set after time $\tau$.

Given a seed set $S_0$, the \emph{influence spread} of $S_0$ is defined as $\sigma(S_0)=\E[|S_{n-1}|]$, i.e.~the expected number of active nodes by the end of the diffusion. Here, $\sigma:2^V\rightarrow\mathbb{R}_+$ is called the \emph{influence spread function}.
In this paper, we also use $\sigma(S,p)$ to state the edge probability vector $p$ explicitly.
The influence maximization (IM) problem takes as input the social network $G$ and an integer $K\in\mathbb{N}_+$, and requires to find the seed set $S^{\opt}$ that gives the maximum influence spread with at most $K$ seeds, i.e. $S^{\opt} \in \argmax_{S\subseteq V, |S|\leq K} \sigma(S)$.
It is well-known that the IM problem admits a $(1-1/e-\epsilon)$ approximation under the IC model \cite{KempeKT03}, which is tight assuming P $\neq$ NP \cite{Feige98}.

\subsection{Online Influence Maximization}

In the online influence maximization problem (OIM) considered in this paper, there is an underlying social network $G=(V,E)$, whose edge parameter vector $p^*$ is unknown initially.
At each round $t$ of total $T$ rounds, the learner chooses a seed set $S_t$ with cardinality at most $K$, observes the cascade feedback, and updates her knowledge about the parameter $p^*$ for later selections.
The feedback considered in this paper is node-level feedback, which means that the learner observes a realization of the sequence of active nodes $(S_{t,0},S_{t,1},\cdots,S_{t,n-1})$ after selecting $S_{t,0}=S_t$.

In order to solve OIM problem, oracles for offline IM problem are often invoked.
Such an oracle takes as input the edge probability vector and outputs a good approximate solution for IM problem.
However, when node-level feedback is used, both \citet{LiKTLC20} and this paper can only guarantee that the true edge probability vector falls into a confidence region.
Thus, we need the so-called \emph{pair oracle} which takes as input the confidence region and can still find a good solution.

Formally, denote by $p^*$ the true edge probability vector and by $\mathcal{C}\in\mathbb{R}^m$ a confidence region satisfying $p^*\in\mathcal{C}$.
Let $\mathtt{ORACLE}$ be a pair-oracle which solves the problem $\max_{S:|S|\leq K,p\in\mathcal{C}} \sigma(S,p)$ and $(\widetilde{S},\tilde{p})=\mathtt{ORACLE}(G,K,\mathcal{C})$ be its output.
Define $S^{\opt}\in\argmax_{S:|S|\leq K} \sigma(S,p^*)$ to be the optimal seed set.
For $\alpha,\beta\in[0,1]$, we say $\mathtt{ORACLE}$ is an $(\alpha,\beta)$-pair-oracle if $\Pr[\sigma(\widetilde{S},\tilde{p})\geq \alpha\cdot \sigma(S^{\opt},p^*)]\geq\beta$,
where the probability is taking from the possible randomness of $\mathtt{ORACLE}$.
Note that $\mathtt{ORACLE}$ is hard to implement, but this paper mainly focuses on the effectiveness of OIM algorithms and the efficiency is not our concern.

Equipped with an $(\alpha,\beta)$-pair-oracle, the objective of OIM is to minimize the cumulative $(\alpha\beta)$-scaled regret over $T$ rounds:
\begin{align*}
	R(T) &=\E\left[\sum_{t=1}^{T}R_t\right] \\
	&=\E\left[T\alpha\beta\cdot\sigma(S^{\opt},p^*)-\sum_{t=1}^{T}|S_{t,n-1}|\right].
\end{align*}
Due to the additivity of expectation, it is equal to
\[ R(T)=\E\left[T\alpha\beta\cdot\sigma(S^{\opt},p^*)-\sum_{t=1}^{T}\sigma(S_t,p^*)\right]. \]

\section{OIM Algorithm under the IC Model}
\label{section: OIM under IC}

In this section, we present an algorithm for OIM under the IC model with node-level feedback (Algorithm \ref{algo: mle-ucb}).
Our algorithm adopts the canonical \emph{upper confidence bound} (UCB) framework in the bandits problem.
Under the UCB framework, at each round $t$, we first compute an estimate $\hat{p}_{t-1}$ of $p^*$ and a corresponding confidence region $\mathcal{C}_{t-1}$  based on the feedback before round $t$.
Then, a seed set $S_t$ is selected by invoking an $(\alpha,\beta)$-pair-oracle to obtain $(S_t,\tilde{p}_{t-1})$, which satisfies that $\tilde{p}_{t-1}\in \mathcal{C}_{t-1}$ and $|S_{t}|\leq K$.

For OIM with node-level feedback, the key difficulty of applying the UCB framework lies in how to use the node-level feedback collected in the previous rounds to update the estimate of $p^*$.
For each node $v\in V$, Algorithm \ref{algo: mle-ucb} will estimate the probability vector $p^*_v\in[0,1]^{d_v}$ of the incoming edges of $v$ separately.
Note that all $p^*_v$ together form $p^*$.

We first explain how to extract information on $p^*_v$ from the feedback $(S_{t,0},S_{t,1},\ldots,S_{t,n-1})$ at round $t$.
When $t>T_0$, the data is processed in a more economical way.
Assume that node $v$ remains inactive after time $\tau$ and some of its neighbors $(S_{t,\tau}\setminus S_{t,\tau-1})\cap N(v)\neq \emptyset$ was newly activated in time $\tau$.
Then, these neighbors will try to activate node $v$ in time $\tau+1$.
Let
\[ E' \coloneqq\{e_{uv}\in E_v\mid u\in (S_{t,\tau}\setminus S_{t,\tau-1})\cap N(v)\} \]
be the set of edges which point from these neighbors to node $v$.
By the diffusion rule of the IC model, the probability that $v$ is activated by them in time $\tau+1$ is
\[ 1-\prod_{e\in E'}(1-p^*(e)). \]
If node $v$ did become active in time $\tau+1$, then we use data pair $(\chi(E'),1)$ to record this event.
Otherwise, we use data pair $(\chi(E'),0)$ to record the event that $v$ remained inactive in time $\tau+1$.
By inspecting each time step of the diffusion till $v$ became active or no new neighbors of $v$ were activated, we are able to construct $J_{t,v}$ data pairs accordingly, denoted by $(X_{t,j,v},Y_{t,j,v})$, $1\leq j\leq J_{t,v}$.
Here, $J_{t,v}\leq d_v$, since $v$ has $d_v$ neighbors and a new data pair is constructed only when some inactive neighbors of $v$ become active.
$X_{t,j,v}\in\{0,1\}^{d_v}$ indicates the characteristic vector of the edges corresponding to the $j$-th batch of neighbors that were activated.
$Y_{t,j,v}\in\{0,1\}$ indicates if $v$ was activated by these neighbors.
It is easy to see when $j<J_{t,v}$, $Y_{t,j,v}=0$, and only when $j=J_{t,v}$, it is possible for $Y_{t,j,v}$ to be $1$.
This is because $v$ will remain active once it is activated.
Though some newly active neighbors of $v$ will still try to ``activate'' $v$ thereafter, it is impossible to observe whether the attempt succeeds.

\begin{algorithm}[tb]
	\caption{IC-UCB}
	\label{algo: mle-ucb}
	\textbf{Input}: Graph $G=(V,E)$, seed set cardinality $K\in\mathbb{N}$, $(\alpha,\beta)$-pair-oracle $\mathtt{ORACLE}$, parameter $\gamma\in(0,1)$ in Assumption \ref{assump: assumption on p*}.
	\begin{algorithmic}[1]
		\STATE Initialize $M_{0,v}\leftarrow \mathbf{0}\in\mathbb{R}^{d_v\times d_v}$ for all $v\in V$, $\delta\leftarrow1/(3n\sqrt{T})$, $R\leftarrow\left\lceil\frac{512D}{\gamma^4}\left(D^2+\ln(1/\delta)\right)\right\rceil$, $T_0\leftarrow nR$ and $\rho\leftarrow\frac{3}{\gamma}\sqrt{\ln(1/\delta)}$.
		\FORALL{$u\in V$}\label{line:initbegin}
		\STATE Choose $\{u\}$ as the seed set for the next $R$ rounds and construct data pairs from observations (see the text in this section for details). \label{line: data pair 1}
		\ENDFOR\label{line:initend}
		\FOR{$t=T_0+1,T_0+2,\cdots,T$}
		\STATE\label{line:update} $\{\hat{\theta}_{t-1,v},\mathcal{C}'_{t-1,v}\}_{v\in V}=$\\Estimate$((S_{k,0},S_{k,1},\ldots, S_{k,n-1})_{1\leq k\leq t-1})$ (see Algorithm~\ref{alg:update}).            
		\STATE Let $\mathcal{C}_{t-1,v}=\{p_v\in[0,1]^{d_v}\mid \theta_v\in\mathcal{C}'_{t-1,v} \}$ and $\mathcal{C}_{t-1}=\{\mathcal{C}_{t-1,v}\}_{v\in V}$
		\STATE Choose $(S_t,\tilde{p}_{t})\in\mathtt{ORACLE}(G,K,\mathcal{C}_{t-1})$ and observe node-level feedback $(S_{t,0},S_{t,1},\cdots, S_{t,n-1})$.
		\ENDFOR
	\end{algorithmic}
\end{algorithm}


\begin{algorithm}[tb]
	\caption{Estimate. Note that the code is written as a computation from scratch in each round
		to accommodate the initialization period of Algorithm~\ref{algo: mle-ucb}, and it can be easily adapted to the incremental computation form.}
	\label{alg:update}
	\textbf{Input}: All observations $(S_{k,0},S_{k,1},\cdots, S_{k,n-1})_{1\leq k\leq t}$ until round $t$.
	\begin{algorithmic}[1]
		\FORALL{$v\in V$}
		\STATE Construct data pairs $(X_{k,j,v},Y_{k,j,v})_{1\leq k\leq t,1\leq j\leq J_{k,v}}$ from observations $(S_{k,0},S_{k,1},\cdots, S_{k,n-1})_{1\leq k\leq t}$
		(see the text in this section for details).	\label{line: data pair 2}
		\STATE \label{line:ltv}$L_{t,v}(\theta_v) \leftarrow \sum_{k=1}^{t}\sum_{j=1}^{J_{k,v}}[-\exp(-X^{\top}_{k,j,v}\theta_v)-(1-Y_{k,j,v})X^{\top}_{k,j,v}\theta_v]$, see eq.~\eqref{eq: pseudo likelihood}.
		\STATE \label{line:theta} $\hat{\theta}_{t,v}\leftarrow \argmax_{\theta_v} L_{t,v}(\theta_v)$.
		\STATE $ M_{t,v}\leftarrow \sum_{k=1}^{t}\sum_{j=1}^{J_{k,v}}X_{k,j,v}X_{k,j,v}^{\top}$, see eq.~\eqref{eq: data matrix}.
		\STATE $\mathcal{C}'_{t,v}\leftarrow\{\theta\in[0,1]^{d_v}\mid \|\theta-\hat{\theta}_{t,v}\|_{M_{t,v}}\leq\rho\}$.
		\ENDFOR
	\end{algorithmic}
\end{algorithm}

For the initial regularization phase (line \ref{line:initbegin} to line \ref{line:initend}) where $t\leq T_0$, the process of extracting information is wasteful in that only the first-step activation is taken into account.
In this part, the algorithm chooses each node $u\in V$ as the seed set for $R$ rounds, and then observes the activation of $u$'s all out-neighbors so as to gather information about its outgoing edges.
More formally, let node $u$ be chosen as the seed in round $t$. In the case $u\in N(v)$, $J_{t,v}=1$ and we construct data pair $(\chi(e_{uv}),1)$ if $v\in S_{t,1}$, or data pair $(\chi(e_{uv}),0)$ if $v\notin S_{t,1}$. In the case $u\notin N(v)$, no data pair is constructed.
By the regularization step, each edge will be observed exactly $R$ times. Intuitively, this step leads to a coarse estimate of each individual probability $p(e)$ for $e\in E$. Technically, this step guarantees a lower bound of the minimum eigenvalue of the Gram matrix $M_{t,v}$ defined in eq.~\eqref{eq: data matrix}, which ensures the correctness of condition \eqref{eq: regular cond} in Theorem \ref{thm: mle} in the analysis.

For each node $v\in V$, we can use all the feedback data $(S_{k,0},S_{k,1},\ldots, S_{k,n-1})_{1\leq k\leq t}$ in the first $t$ rounds to construct data pairs $\{(X_{k,j,v},Y_{k,j,v})\}_{1\leq k\leq t,1\leq j\leq J_{t,v}}$.
The Gram matrix $M_{t,v}$ of these data pairs is defined as
\begin{equation}
	\label{eq: data matrix}
	M_{t,v}\coloneqq \sum_{k=1}^{t}\sum_{j=1}^{J_{k,v}}X_{k,j,v}X_{k,j,v}^{\top}
\end{equation}

Now that we have explained how to extract information on $p^*_v$ by constructing new data pairs from the node-level feedback, we next introduce how to use these data pairs to estimate $p^*_v$.
Inspired by the network inference problem \citep{NetrapalliS12,NarasimhanPS15,Pouget-AbadieH15}, we use the maximum likelihood estimation (MLE) to estimate $p^*_v$.
Algorithm \ref{alg:update} provide a detailed estimation procedure, which has two important features described below.

\emph{Transformation of edge parameter $p$ into parameter $\theta$}.

By the diffusion rule of the IC model, for each $v\in V$, given $X\in\{0,1\}^{d_v}$, let $Y\in\{0,1\}$ indicates whether $v$ is activated in \emph{one time step}. Then,
\[ \E[Y\mid X]=1-\Pi_{e:X(e)=1}(1-p(e)), \]
which a complex function of parameter $p(e)$. We therefore consider a transformation of edge probability vector $p$ into a new vector $\theta$ where
\begin{equation}
	\theta(e)=-\ln(1-p(e)) \mbox{ for each } e\in E.
	\label{eq: para trans}
\end{equation}
Then,
\begin{equation}
	\label{eq: para trans rev}
	p(e)=1-\exp(-\theta(e)) \mbox{ for each } e\in E,
\end{equation}
and
\[ \E[Y\mid X]=\mu(X^{\top}\theta_v), \]
where the \emph{link function} $\mu:\mathbb{R}\rightarrow\mathbb{R}$ is defined as
\[ \mu(x)\coloneqq 1-\exp(-x). \]

This indeed forms an instance of the \emph{generalized linear bandit} (GLB) problem studied in \cite{FilippiCGS10,LiLZ17}.
They also use MLE to solve the GLB problem.
Hence, we will analyze the regret of Algorithm \ref{algo: mle-ucb} via their methods.

\emph{Pseudo log-likelihood function $L_{t,v}$.}

During the update of the estimate of $p^*$ (or $\theta^*$), a standard log-likelihood function is often used:
\begin{align*}
	\mathcal{L}^{std}_{t,v}(\theta_v) &=\sum_{k=1}^{t}\sum_{j=1}^{J_{k,v}}[Y_{k,j,v}\ln\mu(X^{\top}_{k,j,v}\theta_v) \\
	&+(1-Y_{k,j,v})\ln(1-\mu(X^{\top}_{k,j,v}\theta_v))].
\end{align*}
However, the analysis in \cite{FilippiCGS10,LiLZ17} requires that the gradient of the log-likelihood function has the form
\begin{equation}
	\label{eq: gradient}
	\sum_{k=1}^{t}\sum_{j=1}^{J_{k,v}}[Y_{k,j,v}-\mu(X^{\top}_{k,j,v}\theta_v)]X_{k,j,v}.
\end{equation}
Such requirement is met in \citet{FilippiCGS10,LiLZ17} by assuming the distribution of $Y$ conditioned on $X$ falls into some sub-class of the exponential family of distributions, which is however not satisfied in our case.
In this paper, we present an alternative way to overcome such technical difficulty.
That is, we ``integrate'' the gradient in eq.~\eqref{eq: gradient} to obtain a \emph{pseudo log-likelihood} function $L_{t,v}$:
\begin{equation}
	\label{eq: pseudo likelihood}
	\begin{split}
		&L_{t,v}(\theta_v) \\
		&= \sum_{k=1}^{t}\sum_{j=1}^{J_{k,v}}[-\exp(-X^{\top}_{k,j,v}\theta_v)-(1-Y_{k,j,v})X^{\top}_{k,j,v}\theta_v].
	\end{split}
\end{equation}
This ensures that the gradient of $L_{t,v}$ has the form of eq.~\eqref{eq: gradient} and therefore the analysis of \citet{FilippiCGS10,LiLZ17} can be used.
Such an approach is of great independent interest and we leave it as an open problem to find a more intuitive explanation for it.

\section{Regret Analysis}

We now give an analysis of the regret of Algorithm \ref{algo: mle-ucb}.
First, we need to show that for each $v\in V$, the estimate $\hat{\theta}_{t,v}$ is close to the true parameter $\theta^*_v$.
To ensure this, we require Assumption \ref{assump: assumption on p*} below.
\begin{assumption}
	\label{assump: assumption on p*}
	There exists a parameter $\gamma\in (0,1)$ such that $\prod_{u\in N(v)}(1-p^*(e_{uv}))\geq\gamma$ for all $v\in V$.
\end{assumption}
Similar or even stronger assumptions are adopted in all previous approaches for network inference \cite{NetrapalliS12,NarasimhanPS15,Pouget-AbadieH15,ChenSZZ21}.
Assumption \ref{assump: assumption on p*} means that node $v\in V$ will remain inactive with probability at least $\gamma$ even if all of its in-neighbors are simultaneously activated.
It reflects the stubbornness of the agent (node). That is, the behavior of a node is partially determined by its 
intrinsic motivation, not by its neighbors.
So, even when all its neighbors adopt a new behavior, there is a nontrivial 
probability that the node will still not adopt the new behavior.

Under Assumption \ref{assump: assumption on p*}, it is possible to show that $\hat{\theta}_{t,v}$ and $\theta^*_v$ are close to each other in all directions, from which we can obtain a confidence region for $\theta_v^*$, as Theorem \ref{thm: mle} states.
The proof of Theorem \ref{thm: mle} is similar to that of Theorem 1 in \cite{LiLZ17}.
For completeness, we include the proof in
\OnlyInFull{Appendix \ref{section: proof of thm: mle}.}\OnlyInShort{the appendix.}

\begin{theorem}
	\label{thm: mle}
	Suppose that Assumption \ref{assump: assumption on p*} holds.
	For each $v\in V$, $\hat{\theta}_{t,v}$ and $M_{t,v}$ are computed according to Algorithm~\ref{alg:update}.
	Given $\delta\in(0,1)$, if
	\begin{equation}
		\label{eq: regular cond}
		\lambda_{\min}(M_{t,v})\geq\frac{512d_v}{\gamma^4}\left(d_v^2+\ln\frac{1}{\delta}\right).
	\end{equation}
	Then, with probability at least $1-3\delta$, for any $x\in\mathbb{R}^{d_v}$, we have
	\[ |x^{\top}(\hat{\theta}_{t,v}-\theta_v^*)|\leq\frac{3}{\gamma}\sqrt{\ln(1/\delta)}\cdot\|x\|_{M_{t,v}^{-1}}. \]
	Thus, by setting $x=M_{t,v}^{\top}(\hat{\theta}_{t,v}-\theta^*_v)$, we obtain
	\[ \|\hat{\theta}_{t,v}-\theta^*_v\|_{M_{t,v}}\leq \frac{3}{\gamma}\sqrt{\ln(1/\delta)}. \]
\end{theorem}

After we prove that $\hat{\theta}_{t,v}$ and $\theta^*_v$ are indeed close to each other, we need to show that the influence functions $\sigma(S,\hat{p})$ and $\sigma(S,p^*)$ induced by the corresponding probability vectors $\hat{p}_{t}$ and $p^*$ are also close.
To this end, we prove the \emph{group observation modulated} (GOM) bounded smoothness condition for the IC model.
The condition is inspired by the GOM condition for the LT model \cite{LiKTLC20}.
We remark that for edge-level feedback, there is a related \emph{triggering probability modulated} (TPM) bounded smoothness condition \cite{WangC17,WenKVV17}.
However, the TPM condition does not suffice for node-level feedback.

We now state the GOM condition formally.
Given a seed set $S\subseteq V$ and a node $v\in V\setminus S$, we say node $u\in V\setminus S$ is \emph{relevant} to node $v$ if there is a path $P$ from $S$ to $v$ such that $u\in P$.
Let $V[S,v]\subseteq V$ be the set of nodes relevant to $v$ given seed set $S$.
Given diffusion cascade $(S_0=S, S_1,\cdots, S_{n-1})$, construct data pairs $\{(X_{j,v},Y_{j,v})\}_{1\leq j\leq J_v}$ according to the economical way described previously (not the wasteful way for the regularization phase).
We have the following GOM condition for the IC model, whose proof is presented in
\OnlyInFull{Appendix \ref{section: proof of lemma: TPM}.}\OnlyInShort{the appendix.}

\begin{lemma}[GOM bounded smoothness for the IC model]
	\label{lemma: TPM}
	Fix any seed set $S\subseteq V$.
	For any two edge-probability vectors $\tilde{p},\,p^*\in[0,1]^{|E|}$, let $\tilde{\theta},\,\theta^*$ be the vectors defined as eq.~\eqref{eq: para trans}.
	Then,
	\begin{align*}
		&|\sigma(S, \tilde{p})-\sigma(S,p^*)| \\
		&\leq\sum_{v\in V\setminus S}\sum_{u\in V[S,v]}\E\left[\sum_{j=1}^{J_{u}}\left|X_{j,u}^{\top}(\tilde{\theta}_u-\theta^*_u)\right|\right],
	\end{align*}
    where the expectation is taken over the randomness of the diffusion cascade $(S_0,S_1,\cdots, S_{n-1})$, which is generated with respect to parameter $p^*$.
\end{lemma}

Equipped with the aforementioned tools, we now set about presenting the analysis of Algorithm \ref{algo: mle-ucb}.
Given a seed set $S\subseteq V$ and a node $u\in V\setminus S$, define
\[ n_{S,u}\coloneqq\sum_{v\in V\setminus S}\mathbf{1}\{u\in V[S,v]\} \]
to be the number of nodes that $u$ is relevant to.
Further, define
\[ \zeta(G)\coloneqq \max_{S:|S|\leq K}\sqrt{\sum_{u\in V}n_{S,u}^2}\leq O(n^{3/2}). \]
We present the regret of Algorithm \ref{algo: mle-ucb} in Theorem \ref{thm: regret2-IC}.

\begin{theorem}
	\label{thm: regret2-IC}
    When we use an $(\alpha,\beta)$-pair-oracle in Algorithm \ref{algo: mle-ucb}, under Assumption \ref{assump: assumption on p*}, the $\alpha\beta$-scaled regret of Algorithm \ref{algo: mle-ucb} satisfies that
	\begin{align*}
		R(T)= \widetilde{O}\left(\frac{\zeta(G)D\sqrt{mT}}{\gamma}\right)=\widetilde{O}\left(\frac{n^{7/2}\sqrt{T}}{\gamma}\right).
	\end{align*}
\end{theorem}

\begin{proof}
	Let $\mathcal{H}_t$ be the history of past rounds by the end of round $t$.
	For $t\leq T_0$, $\E[R_t]\leq n$, since there are $n$ nodes in $G$.
	Now consider the case where $t>T_0$.
	By the definition of $R_t$,
	\[ \E[R_t\mid\mathcal{H}_{t-1}]=\E[\alpha\beta\cdot\sigma(S^{\opt},p^*)-\sigma(S_t,p^*)\mid\mathcal{H}_{t-1}], \]
	where the expectation is taken over the randomness of $S_t$.
	
	For any $T_0<t\leq T$ and $v\in V$, define event $\xi_{t-1,v}$ as
	\[ \xi_{t-1,v}\coloneqq\{\|\hat{\theta}_{t-1,v}-\theta^*_v\|_{M_{t-1,v}}\leq \rho \}, \]
	and let $\overline{\xi}_{t-1,v}$ be its complement.
	By the choices of $\delta,R,T_0,\rho$ as in Algorithm \ref{algo: mle-ucb}, the fact that $\lambda_{\min}(M_{t-1,v})\geq \lambda_{\min}(M_{T_0,v})=R$ and Theorem \ref{thm: mle}, we have $\Pr[\overline{\xi}_{t-1,v}]\leq 3\delta$.
	Further define event $\xi_{t-1}\coloneqq\wedge_{v\in V}\,\xi_{t-1,v}$ and let $\overline{\xi}_{t-1}$ be its complement.
	By union bound, $\Pr[\overline{\xi}_{t-1}]\leq 3\delta n$.
	Note that under event $\xi_{t-1}$, for all $v\in V$, $\theta^*_v\in\mathcal{C}'_{t-1,v}$.
	Hence, $p^*_{v}\in \mathcal{C}_{t-1,v}$ and $p^*\in \mathcal{C}_{t-1}$.
	Since $(S_t,\tilde{p}_t)$ is obtained by invoking an $(\alpha,\beta)$-pair-oracle $\mathtt{ORACLE}$ over $\mathcal{C}_{t-1}$, we have
	\begin{align*}
		\E[R_t] &\leq \Pr[\xi_{t-1}]\cdot\E[\alpha\beta\cdot\sigma(S^{\opt}, p^*)-\sigma(S_t,p^*)\mid\xi_{t-1}] \\
		&+\Pr[\overline{\xi}_{t-1}]\cdot n \\
		&\leq \E[\sigma(S_t, \tilde{p}_t)-\sigma(S_t,p^*)\mid\xi_{t-1}]+3\delta n^2.
	\end{align*}
    
    Next, by the GOM bounded smoothness for the IC model in Lemma \ref{lemma: TPM}, we obtain that
    \begin{align*}
    	&\E[R_t]-3\delta n^2 \\
    	&\leq \E\left[\sum_{v\in V\setminus S_t}\sum_{u\in V[S_t,v]}\sum_{j=1}^{J_{t,u}}\bigg|X_{t,j,u}^{\top}(\tilde{\theta}_{t,u}-\theta^*_u)\bigg|\,\bigg|\,\xi_{t-1}\right].
    \end{align*}
    By the Cauchy-Schwarz inequality, we have
    \[ |X_{t,j,u}^{\top}(\tilde{\theta}_{t,u}-\theta^*_u)|\leq \|X_{t,j,u}\|_{M_{t-1,u}^{-1}}\|\tilde{\theta}_{t,u}-\theta^*_u\|_{M_{t-1,u}}. \]
    Besides, under event $\xi_{t-1}$, $\hat{\theta}_{t,u},\theta^*_u\in\mathcal{C}'_{t-1,u}$ for all $u\in V$.
    Then, by the triangle inequality,
    \begin{align*}
    	&\|\tilde{\theta}_{t,u}-\theta^*_u\|_{M_{t-1,u}} \\
    	&\leq \|\tilde{\theta}_{t,u}-\hat{\theta}_{t-1,u}\|_{M_{t-1,u}}+\|\hat{\theta}_{t-1,u}-\theta^*_u\|_{M_{t-1,u}}\leq 2\rho.
    \end{align*}
    Combining the above inequalities, we obtain that
    \begin{align*}
    	&\E[R_t]-3\delta n^2 \\
    	    	&\leq 2\rho\cdot\E\left[\sum_{v\in V\setminus S_t}\sum_{u\in V[S_t,v]}\sum_{j=1}^{J_{t,u}}\|X_{t,j,u}\|_{M_{t-1,u}^{-1}}\right] \\
    	    	&=2\rho\cdot\E\left[\sum_{u\in V\setminus S_t}\sum_{j=1}^{J_{t,u}}\|X_{t,j,u}\|_{M_{t-1,u}^{-1}}\sum_{v\in V\setminus S_t}\mathbf{1}_{u\in V[S_t,v]}\right] \\
    	    	&= 2\rho\cdot \E\left[\sum_{u\in V\setminus S_t} n_{S_t,u}\sum_{j=1}^{J_{t,u}}\|X_{t,j,u}\|_{M_{t-1,u}^{-1}}\right].
    \end{align*}
	
	Recall that the above derivation holds for $t>T_0$, and for $t\leq T_0$, $\E[R_t]\leq n$.
	We thus have
	\begin{align*}
		R(T) &\leq 2\rho\cdot\E\left[\sum_{t=T_0+1}^{T}\sum_{v\in V\setminus S_t}n_{S_t,v}\sum_{j=1}^{J_{t,v}}\|X_{t,j,v}\|_{M_{t-1,v}^{-1}}\right] \\
		&+3\delta n^2(T-T_0)+nT_0.
	\end{align*}

    To further simplify the above inequality, we prove the following lemma, whose proof is presented in \OnlyInFull{Appendix \ref{section: proof of lemma: upper bound of modified M norm}.}\OnlyInShort{the appendix.}
    \begin{lemma}
    	\label{lemma: upper bound of modified M norm}
    	For any $v\in V$,
    	\begin{align*}
    		&\sum_{t=T_0+1}^{T}\sum_{v\in V\setminus S_t}n_{S_t,v}\sum_{j=1}^{J_{t,v}}\|X_{t,j,v}\|_{M_{t-1,v}^{-1}} \\
    		&\leq \zeta(G)D\sqrt{(m+n)(T-T_0)\ln\left(R+(T-T_0)D\right)}.
    	\end{align*}
    \end{lemma}
    By Lemma \ref{lemma: upper bound of modified M norm}, we have
    \begin{align*}
    	&R(T) \\
    	&\leq 2\rho\cdot\E\left[\sum_{t=T_0+1}^{T}\sum_{v\in V\setminus S_t}n_{S_t,v}\sum_{j=1}^{J_{t,v}}\|X_{t,j,v}\|_{M_{t-1,v}^{-1}}\right] \\
    	&+3\delta n^2(T-T_0)+nT_0 \\
    	&\leq 2\rho \zeta(G)D\sqrt{(m+n)(T-T_0)\ln\left(R+(T-T_0)D\right)} \\
    	&+3\delta n^2(T-T_0)+nT_0 \\
    	&\leq \frac{6 \zeta(G)D}{\gamma}\sqrt{(m+n)T\ln(TD)\ln (3nT)}+n\sqrt{T} \\
    	&+\frac{512Dn^2}{\gamma^4}\left(D^2+\ln(3nT)\right)+1 \\
    	&=\widetilde{O}\left(\frac{\zeta(G)D\sqrt{mT}}{\gamma}\right).
    \end{align*}
    The last inequality is obtained by plugging $\delta=1/(3n\sqrt{T})$, $R=\left\lceil\frac{512D}{\gamma^4}\left(D^2+\ln(1/\delta)\right)\right\rceil$, $T_0= nR$ and $\rho=\frac{3}{\gamma}\sqrt{\ln(1/\delta)}$ into the formula.

\end{proof}

We remark that the worst-case regret for the IC model with edge-level feedback is $\widetilde{O}(n^3\sqrt{T})$ in \cite{WangC17}.
Thus, our regret bound under node-level feedback matches the previous ones under edge-level feedback in the worst case, up to a $n^{1/2}/\gamma$ factor.

To get an intuition about $\gamma$'s value, assume that each edge probability $\leq 1-c$ for some constant $c\in(0,1)$.
Then, $\gamma=O(c^D)$, where $D$ is the maximum in-degree of the graph.
Thus, in the worst case, $1/\gamma$ is exponential in $n$.
But when $D=O(\log n)$, $1/\gamma$ is polynomial in $n$ and so is the regret bound.
We think $D=O(\log n)$ is reasonable in practice, since a person only has a limited attention and cannot pay attention to too many
people in the network.

\section{Conclusion}

In this paper, we investigate the OIM problem under the IC model with node-level feedback.
We presents an $\widetilde{O}(\sqrt{T})$-regret OIM algorithm for the problem, which almost matches the optimal regret bound as well as the state-of-the-art regret bound with edge-level feedback.
Our novel adaptation of MLE to fit the GLB model is of great independent interest, which might be combined with the GLB model to handle rewards generated from a broader classes of distributions.

There still remain several open problems in the node-level feedback setting.
An immediate one is to either remove Assumption \ref{assump: assumption on p*} for the edge probability vector or at least the assumption parameter from the regret bound.
Besides, one can also study if there exists optimal-regret OIM algorithms with node-level feedback that use standard offline oracles.
Finally, it is interesting to develop a general bandit framework which includes OIM with node-level feedback as a special case, just like CMAB-T containing OIM with edge-level feedback.

\section*{Acknowledgements}

This work was supported in part by the National Natural Science Foundation of China Grants No.~61832003, 61872334, the Strategic Priority Research Program of Chinese Academy of Sciences under Grant No.~XDA27000000.

\bibliography{OIM}

\begin{thebibliography}{40}
\providecommand{\natexlab}[1]{#1}

\bibitem[{Abbasi{-}Yadkori, P{\'{a}}l, and
  Szepesv{\'{a}}ri(2011)}]{Abbasi-YadkoriPS11}
Abbasi{-}Yadkori, Y.; P{\'{a}}l, D.; and Szepesv{\'{a}}ri, C. 2011.
\newblock Improved Algorithms for Linear Stochastic Bandits.
\newblock In \emph{Advances in Neural Information Processing Systems 24: 25th
  Annual Conference on Neural Information Processing Systems 2011.},
  2312--2320.

\bibitem[{Abrahao et~al.(2013)Abrahao, Chierichetti, Kleinberg, and
  Panconesi}]{AbrahaoCKP13}
Abrahao, B.~D.; Chierichetti, F.; Kleinberg, R.; and Panconesi, A. 2013.
\newblock Trace complexity of network inference.
\newblock In \emph{the 19th {ACM} {SIGKDD} International Conference on
  Knowledge Discovery and Data Mining, {KDD} 2013}, 491--499. {ACM}.

\bibitem[{Auer, Cesa{-}Bianchi, and Fischer(2002)}]{AuerCF02}
Auer, P.; Cesa{-}Bianchi, N.; and Fischer, P. 2002.
\newblock Finite-time Analysis of the Multiarmed Bandit Problem.
\newblock \emph{Mach. Learn.}, 47(2-3): 235--256.

\bibitem[{Borgs et~al.(2014)Borgs, Brautbar, Chayes, and Lucier}]{BorgsBCL14}
Borgs, C.; Brautbar, M.; Chayes, J.~T.; and Lucier, B. 2014.
\newblock Maximizing Social Influence in Nearly Optimal Time.
\newblock In Chekuri, C., ed., \emph{Proceedings of the Twenty-Fifth Annual
  {ACM-SIAM} Symposium on Discrete Algorithms, {SODA} 2014}, 946--957. {SIAM}.

\bibitem[{Chen, Hu, and Ying(1999)}]{ChenHY99}
Chen, K.; Hu, I.; and Ying, Z. 1999.
\newblock {Strong consistency of maximum quasi-likelihood estimators in
  generalized linear models with fixed and adaptive designs}.
\newblock \emph{The Annals of Statistics}, 27(4): 1155 -- 1163.

\bibitem[{Chen, Lakshmanan, and Castillo(2013)}]{ChenLC13}
Chen, W.; Lakshmanan, L. V.~S.; and Castillo, C. 2013.
\newblock \emph{Information and Influence Propagation in Social Networks}.
\newblock Synthesis Lectures on Data Management. Morgan {\&} Claypool
  Publishers.

\bibitem[{Chen et~al.(2021)Chen, Sun, Zhang, and Zhang}]{ChenSZZ21}
Chen, W.; Sun, X.; Zhang, J.; and Zhang, Z. 2021.
\newblock Network Inference and Influence Maximization from Samples.
\newblock In \emph{Proceedings of the 38th International Conference on Machine
  Learning, {ICML} 2021}, 1707--1716. {PMLR}.

\bibitem[{Chen, Wang, and Yuan(2013)}]{ChenWY13}
Chen, W.; Wang, Y.; and Yuan, Y. 2013.
\newblock Combinatorial Multi-Armed Bandit: General Framework and Applications.
\newblock In \emph{Proceedings of the 30th International Conference on Machine
  Learning, {ICML} 2013}, 151--159. JMLR.

\bibitem[{Chen et~al.(2016)Chen, Wang, Yuan, and Wang}]{ChenWYW16}
Chen, W.; Wang, Y.; Yuan, Y.; and Wang, Q. 2016.
\newblock Combinatorial Multi-Armed Bandit and Its Extension to
  Probabilistically Triggered Arms.
\newblock \emph{J. Mach. Learn. Res.}, 17: 50:1--50:33.

\bibitem[{Daneshmand et~al.(2014)Daneshmand, Gomez{-}Rodriguez, Song, and
  Sch{\"{o}}lkopf}]{DaneshmandGSS14}
Daneshmand, H.; Gomez{-}Rodriguez, M.; Song, L.; and Sch{\"{o}}lkopf, B. 2014.
\newblock Estimating Diffusion Network Structures: Recovery Conditions, Sample
  Complexity {\&} Soft-thresholding Algorithm.
\newblock In \emph{Proceedings of the 30th International Conference on Machine
  Learning, {ICML} 2014}, 793--801. JMLR.

\bibitem[{Dani, Hayes, and Kakade(2008)}]{DaniHK08}
Dani, V.; Hayes, T.~P.; and Kakade, S.~M. 2008.
\newblock Stochastic Linear Optimization under Bandit Feedback.
\newblock In \emph{the 21st Annual Conference on Learning Theory, {COLT} 2008},
  355--366. Omnipress.

\bibitem[{Du et~al.(2014)Du, Liang, Balcan, and Song}]{DuLBS14}
Du, N.; Liang, Y.; Balcan, M.; and Song, L. 2014.
\newblock Influence Function Learning in Information Diffusion Networks.
\newblock In \emph{Proceedings of the 31th International Conference on Machine
  Learning, {ICML} 2014}, 2016--2024. JMLR.

\bibitem[{Du et~al.(2013)Du, Song, Gomez{-}Rodriguez, and Zha}]{DuSGZ13}
Du, N.; Song, L.; Gomez{-}Rodriguez, M.; and Zha, H. 2013.
\newblock Scalable Influence Estimation in Continuous-Time Diffusion Networks.
\newblock In \emph{Advances in Neural Information Processing Systems 26 ({NIPS}
  2013)}, 3147--3155.

\bibitem[{Du et~al.(2012)Du, Song, Smola, and Yuan}]{DuSSY12}
Du, N.; Song, L.; Smola, A.~J.; and Yuan, M. 2012.
\newblock Learning Networks of Heterogeneous Influence.
\newblock In \emph{Advances in Neural Information Processing Systems 25 ({NIPS}
  2012)}, 2789--2797.

\bibitem[{Feige(1998)}]{Feige98}
Feige, U. 1998.
\newblock A Threshold of ln \emph{n} for Approximating Set Cover.
\newblock \emph{J. {ACM}}, 45(4): 634--652.

\bibitem[{Filippi et~al.(2010)Filippi, Capp{\'{e}}, Garivier, and
  Szepesv{\'{a}}ri}]{FilippiCGS10}
Filippi, S.; Capp{\'{e}}, O.; Garivier, A.; and Szepesv{\'{a}}ri, C. 2010.
\newblock Parametric Bandits: The Generalized Linear Case.
\newblock In \emph{Advances in Neural Information Processing Systems 23 ({NIPS}
  2010)}, 586--594. Curran Associates, Inc.

\bibitem[{Gomez{-}Rodriguez, Balduzzi, and
  Sch{\"{o}}lkopf(2011)}]{Gomez-RodriguezBS11}
Gomez{-}Rodriguez, M.; Balduzzi, D.; and Sch{\"{o}}lkopf, B. 2011.
\newblock Uncovering the Temporal Dynamics of Diffusion Networks.
\newblock In \emph{Proceedings of the 28th International Conference on Machine
  Learning, {ICML} 2011}, 561--568. Omnipress.

\bibitem[{Gomez{-}Rodriguez, Leskovec, and Krause(2010)}]{Gomez-RodriguezLK10}
Gomez{-}Rodriguez, M.; Leskovec, J.; and Krause, A. 2010.
\newblock Inferring networks of diffusion and influence.
\newblock In \emph{Proceedings of the 16th {ACM} {SIGKDD} International
  Conference on Knowledge Discovery and Data Mining, {KDD} 2010}, 1019--1028.
  {ACM}.

\bibitem[{He et~al.(2016)He, Xu, Kempe, and Liu}]{HeX0L16}
He, X.; Xu, K.; Kempe, D.; and Liu, Y. 2016.
\newblock Learning Influence Functions from Incomplete Observations.
\newblock In \emph{Advances in Neural Information Processing Systems 29 ({NIPS}
  2016)}, 2065--2073.

\bibitem[{Kempe, Kleinberg, and Tardos(2003)}]{KempeKT03}
Kempe, D.; Kleinberg, J.~M.; and Tardos, {\'{E}}. 2003.
\newblock Maximizing the spread of influence through a social network.
\newblock In \emph{Proceedings of the Ninth {ACM} {SIGKDD} International
  Conference on Knowledge Discovery and Data Mining, {KDD} 2003}, 137--146.
  {ACM}.

\bibitem[{Lai and Robbins(1985)}]{LaiR85}
Lai, T.~L.; and Robbins, H. 1985.
\newblock Asymptotically efficient adaptive allocation rules.
\newblock \emph{Advances in applied mathematics}, 6(1): 4--22.

\bibitem[{Lei et~al.(2015)Lei, Maniu, Mo, Cheng, and Senellart}]{LeiMMCS15}
Lei, S.; Maniu, S.; Mo, L.; Cheng, R.; and Senellart, P. 2015.
\newblock Online Influence Maximization.
\newblock In \emph{Proceedings of the 21th {ACM} {SIGKDD} International
  Conference on Knowledge Discovery and Data Mining, {KDD} 2015}, 645--654.
  {ACM}.

\bibitem[{Li et~al.(2010)Li, Chu, Langford, and Schapire}]{LiCLS10}
Li, L.; Chu, W.; Langford, J.; and Schapire, R.~E. 2010.
\newblock A contextual-bandit approach to personalized news article
  recommendation.
\newblock In \emph{Proceedings of the 19th International Conference on World
  Wide Web, {WWW} 2010}, 661--670. {ACM}.

\bibitem[{Li, Lu, and Zhou(2017)}]{LiLZ17}
Li, L.; Lu, Y.; and Zhou, D. 2017.
\newblock Provably Optimal Algorithms for Generalized Linear Contextual
  Bandits.
\newblock In \emph{Proceedings of the 34th International Conference on Machine
  Learning, {ICML} 2017}, 2071--2080. {PMLR}.

\bibitem[{Li et~al.(2020)Li, Kong, Tang, Li, and Chen}]{LiKTLC20}
Li, S.; Kong, F.; Tang, K.; Li, Q.; and Chen, W. 2020.
\newblock Online Influence Maximization under Linear Threshold Model.
\newblock In \emph{Advances in Neural Information Processing Systems 33
  (NeurIPS 2020)}, 1192--1204.

\bibitem[{Li et~al.(2018)Li, Fan, Wang, and Tan}]{LiFWT18}
Li, Y.; Fan, J.; Wang, Y.; and Tan, K.-L. 2018.
\newblock Influence maximization on social graphs: A survey.
\newblock \emph{IEEE Transactions on Knowledge and Data Engineering}, 30(10):
  1852--1872.

\bibitem[{Myers and Leskovec(2010)}]{MyersL10}
Myers, S.~A.; and Leskovec, J. 2010.
\newblock On the Convexity of Latent Social Network Inference.
\newblock In \emph{Advances in Neural Information Processing Systems 23 ({NIPS}
  2010)}, 1741--1749. Curran Associates, Inc.

\bibitem[{Narasimhan, Parkes, and Singer(2015)}]{NarasimhanPS15}
Narasimhan, H.; Parkes, D.~C.; and Singer, Y. 2015.
\newblock Learnability of Influence in Networks.
\newblock In \emph{Advances in Neural Information Processing Systems 28 ({NIPS}
  2015)}, 3186--3194.

\bibitem[{Netrapalli and Sanghavi(2012)}]{NetrapalliS12}
Netrapalli, P.; and Sanghavi, S. 2012.
\newblock Learning the graph of epidemic cascades.
\newblock In \emph{{ACM} {SIGMETRICS/PERFORMANCE} Joint International
  Conference on Measurement and Modeling of Computer Systems, {SIGMETRICS}
  2012}, 211--222. {ACM}.

\bibitem[{Pollard(1990)}]{Pollard90}
Pollard, D. 1990.
\newblock Empirical Processes: Theory and Applications.
\newblock \emph{NSF-CBMS Regional Conference Series in Probability and
  Statistics}, 2: i--86.

\bibitem[{Pouget{-}Abadie and Horel(2015)}]{Pouget-AbadieH15}
Pouget{-}Abadie, J.; and Horel, T. 2015.
\newblock Inferring Graphs from Cascades: {A} Sparse Recovery Framework.
\newblock In \emph{Proceedings of the 32nd International Conference on Machine
  Learning, {ICML} 2015}, 977--986. JMLR.

\bibitem[{Robbins(1952)}]{Robbins52}
Robbins, H. 1952.
\newblock Some aspects of the sequential design of experiments.
\newblock \emph{Bulletin of the American Mathematical Society}, 58(5):
  527--535.

\bibitem[{Rusmevichientong and Tsitsiklis(2010)}]{RusmevichientongT10}
Rusmevichientong, P.; and Tsitsiklis, J.~N. 2010.
\newblock Linearly Parameterized Bandits.
\newblock \emph{Math. Oper. Res.}, 35(2): 395--411.

\bibitem[{Tang, Shi, and Xiao(2015)}]{TangSX15}
Tang, Y.; Shi, Y.; and Xiao, X. 2015.
\newblock Influence Maximization in Near-Linear Time: {A} Martingale Approach.
\newblock In \emph{International Conference on Management of Data, {SIGMOD}
  2015}, 1539--1554. {ACM}.

\bibitem[{Tang, Xiao, and Shi(2014)}]{TangXS14}
Tang, Y.; Xiao, X.; and Shi, Y. 2014.
\newblock Influence maximization: near-optimal time complexity meets practical
  efficiency.
\newblock In \emph{International Conference on Management of Data, {SIGMOD}
  2014}, 75--86. {ACM}.

\bibitem[{Thompson(1933)}]{Thompson33}
Thompson, W.~R. 1933.
\newblock On the likelihood that one unknown probability exceeds another in
  view of the evidence of two samples.
\newblock \emph{Biometrika}, 25(3/4): 285--294.

\bibitem[{Vaswani, Lakshmanan, and Schmidt(2016)}]{VaswaniLS15}
Vaswani, S.; Lakshmanan, L. V.~S.; and Schmidt, M. 2016.
\newblock Influence Maximization with Bandits.
\newblock arXiv:1503.00024.

\bibitem[{Wang and Chen(2017)}]{WangC17}
Wang, Q.; and Chen, W. 2017.
\newblock Improving Regret Bounds for Combinatorial Semi-Bandits with
  Probabilistically Triggered Arms and Its Applications.
\newblock In \emph{Advances in Neural Information Processing Systems 30 ({NIPS}
  2017)}, 1161--1171.

\bibitem[{Wen et~al.(2017)Wen, Kveton, Valko, and Vaswani}]{WenKVV17}
Wen, Z.; Kveton, B.; Valko, M.; and Vaswani, S. 2017.
\newblock Online Influence Maximization under Independent Cascade Model with
  Semi-Bandit Feedback.
\newblock In \emph{Advances in Neural Information Processing Systems 30 ({NIPS}
  2017)}, 3022--3032.

\bibitem[{Wu et~al.(2019)Wu, Li, Wang, Chen, and Wang}]{WuLWCW19}
Wu, Q.; Li, Z.; Wang, H.; Chen, W.; and Wang, H. 2019.
\newblock Factorization Bandits for Online Influence Maximization.
\newblock In \emph{Proceedings of the 25th {ACM} {SIGKDD} International
  Conference on Knowledge Discovery {\&} Data Mining, {KDD} 2019}, 636--646.
  {ACM}.

\end{thebibliography}

\OnlyInShort{\end{document}}

\newpage
\onecolumn
\appendix
\section*{Appendix}

\section{Proof of Theorem \ref{thm: mle}}
\label{section: proof of thm: mle}

OIM with node-level feedback under the IC model is closely related to the \emph{generalized linear bandits} (GLB) problem, introduced by \cite{FilippiCGS10} and further investigated in \cite{LiLZ17}.
To see this, we introduce the GLB model.
Assume there are $T$ rounds in total.
At round $t$, an action $X_t\in\mathbb{R}^d$ with $X_t\neq\mathbf{0}$ is chosen.
Assume there is an unknown parameter $\theta^*\in\mathbb{R}^d$ and a \emph{link function} $\mu:\mathbb{R}\rightarrow\mathbb{R}$.
The reward $Y_t\in\mathbb{R}$ given $X_t$ satisfies that
\[ \E[Y_t\mid X_t]=\mu(X^{\top}_t\theta^*). \]
The goal is to minimize the cumulative regret over $T$ rounds.

Specific to our problem, $X_t\in\{0,1\}^d$ and $Y_t\in\{0,1\}$.
The \emph{link function} $\mu:\mathbb{R}\rightarrow\mathbb{R}$ is defined as $\mu(x)=1-\exp(-x)$.
Besides, $\theta^*\in\Theta\coloneqq\{\theta\in\mathbb{R}^d\mid \theta(e)\geq0,\forall e\in[d],\sum_{e\in [d]} \theta(e)\leq\ln(1/\gamma) \}$, which is induced by Assumption \ref{assump: assumption on p*}.
However, our problem is more difficult in that multiple data pairs $(X_{t,j,v},Y_{t,j,v})$ may be generated for each $v\in V$ at each round $t$, and the reward is specified by a more involved influence function.

The following lemma gives some properties about $\mu$, which are easy to verify and useful for our analysis.
\begin{lemma}
	For any $\theta\in\Theta$ and $X\in\{0,1\}^d$ with $X\neq\mathbf{0}$, it satisfies that
	\[ \dot{\mu}(X^{\top}\theta)\leq 1,\,\dot{\mu}(X^{\top}\theta)\geq\gamma, \mbox{ and }|\ddot{\mu}(X^{\top}\theta)|\leq 1. \]
\end{lemma}

%
%
%
%
%
The analysis in \cite{LiLZ17} cannot be applied here directly for $(X_k,Y_k)$, since it requires that the gradient of the log-likelihood function has the form
$\sum_{k=1}^{t}[Y_{k}-\mu(X^{\top}_{k}\theta)]X_{k}$.
We thus define a pseudo log-likelihood function by ``integrating'' the gradient to apply the analysis.
As in the main text, we gain an estimate $\hat{\theta}_t$ of $\theta^*$ at round $t+1$ by maximizing the following pseudo log-likelihood function:
%
%
\begin{equation}
	\hat{\theta}_t=\argmax L_t(\theta),
	\label{eq: mle_glb}
\end{equation}
where
\[ L_t(\theta)=\sum_{k=1}^{t}[-\exp(-X_k^{\top}\theta)-(1-Y_k)X_k^{\top}\theta]. \]

The following theorem characterizes the confidence intervals of $\theta^*$ induced by $\hat{\theta}_t$, which is the same as Theorem \ref{thm: mle}.
We will prove it instead of Theorem \ref{thm: mle}.

\begin{theorem}
	Assume that $\theta^*\in\Theta$.
	Define $M_t=\sum_{k=1}^{t}X_kX_k^{\top}$ and let $\hat{\theta}_t$ be defined as in Eq.~\eqref{eq: mle_glb}.
	Given $\delta\in(0,1)$, assume that
	\[ \lambda_{\min}(M_t)\geq\frac{512d}{\gamma^4}\left(d^2+\ln\frac{1}{\delta}\right). \]
	Then, with probability at least $1-3\delta$, for any $x\in\mathbb{R}^d$, it satisfies that
	\[ |x^{\top}(\hat{\theta}_t-\theta^*)|\leq\frac{3}{\gamma}\sqrt{\ln(1/\delta)}\cdot\|x\|_{M_t^{-1}}. \]
\end{theorem}

\begin{proof}
	The proof consists of two steps.
	In the first step, it is proved that $\hat{\theta}_t$ falls into the $\eta$-neighborhood $\mathcal{B}_{\eta}$ of $\theta^*$ w.r.t.~$\ell_2$-norm for some $\eta$.
	Since $\dot{\mu}(X^{\top}\theta^*)\geq\gamma$, for any $\theta\in\mathcal{B}_{\eta}$, $\dot{\mu}(X^{\top}\theta)$ also has a lower bound denoted by $\kappa_{\eta}$.
	The values of $\eta$ and $\kappa_{\eta}$ will be determined later.
	In the second step, it is proved that $\hat{\theta}$ and $\theta^*$ are close in any direction $x\in\mathbb{R}^d$.
	
	First note that $\hat{\theta}_t$ satisfies that $\nabla L_t(\hat{\theta}_t)=0$, where the gradient $\nabla L_t(\theta)$ is
	\[ \nabla L_t(\theta)=\sum_{k=1}^{t}[\exp(-X_k^{\top}\theta)-(1-Y_k)]X_k=\sum_{k=1}^{t}[Y_k-\mu(X_k^{\top}\theta)]X_k. \]
	Define $G(\theta)\coloneqq\sum_{k=1}^{t}(\mu(X_k^{\top}\theta)-\mu(X_k^{\top}\theta^*))X_k$.
	Then, we have
	\[ G(\theta^*)=0 \mbox{ and } G(\hat{\theta}_t)=\sum_{k=1}^{t}\epsilon_kX_k, \]
	where $\epsilon_k$ is defined as $\epsilon_k\coloneqq Y_k-\mu(X_k^{\top}\theta^*)$.
	Note that $\E[\epsilon_k\mid X_k]=0$ and $\epsilon_k=Y_k-\mu(X_k^{\top}\theta^*)\in[-1,1]$ since $Y_k\in\{0,1\}$ and $\mu(X_k^{\top}\theta^*)=\Pr[Y_k=1\mid X_k]\in[0,1]$.
	Therefore, $\epsilon_k$ is $1$-sub-Gaussian, i.e.~$\E[\exp(\lambda\epsilon_k)\mid X_k]\leq\exp(\lambda^2/2),\forall\,\lambda\in\mathbb{R}$.
	Further, define $Z\coloneqq G(\hat{\theta}_t)=\sum_{k=1}^{t}\epsilon_kX_k$ for convenience.
	
	\paragraph{Step 1: Consistency of $\hat{\theta}_t$.} We first prove the consistency of $\hat{\theta}_t$.
	For any $\theta_1,\theta_2\in\mathbb{R}^d_+\coloneqq\{\theta\in\mathbb{R}^d\mid \theta(e)\geq0,\forall e\in[d]\}$, by the mean value theorem, there is some $\bar{\theta}=s\theta_1+(1-s)\theta_2$ with $0<s<1$ such that
	\[ G(\theta_1)-G(\theta_2)=\left[\sum_{k=1}^{t}\dot{\mu}(X_k^{\top}\bar{\theta})X_kX_k^{\top}\right](\theta_1-\theta_2)\coloneqq F(\bar{\theta})(\theta_1-\theta_2). \]
	Since $\dot{\mu}(x)=\exp(-x)$, for $\bar{\theta}\in\mathbb{R}^d_+$, $\dot{\mu}(X^{\top}\bar{\theta})>0$.
	Together with $\lambda_{\min}(M_t)>0$, we have $\lambda_{\min}(F(\bar{\theta}))>0$.
	Therefore, for any $\theta_1\neq\theta_2$,
	\[ (\theta_1-\theta_2)^{\top}(G(\theta_1)-G(\theta_2))=(\theta_1-\theta_2)^{\top}F(\bar{\theta})(\theta_1-\theta_2)>0. \]
	Consequently, $G(\theta)$ is an injection from $\mathbb{R}^d$ to $\mathbb{R}^d$ and therefore $G^{-1}$ is well-defined.
	We thus have $\hat{\theta}=G^{-1}(Z)$.
	
	Let $\mathcal{B}_{\eta}\coloneqq\{\theta\mid \|\theta-\theta^*\|\leq\eta\}$ be the $\eta$-neighborhood of $\theta^*$ and $\partial\mathcal{B}_{\eta}\coloneqq\{\theta\mid \|\theta-\theta^*\|=\eta\}$.
	Define $\kappa_{\eta}\coloneqq\inf_{\theta\in\mathcal{B}_{\eta},X\neq\mathbf{0}}\dot{\mu}(X^{\top}\theta)>0$.
	The following lemma shows that if $G(\theta)$ and $G(\theta^*)$ are close, then $\theta$ and $\theta^*$ are also close.
	Its proof is presented in Appendix \ref{section: proof of lemma: G(theta)'s bound 1}.
	
	\begin{lemma}
		\label{lemma: G(theta)'s bound 1}
		$\{\theta\mid \|G(\theta)\|_{M_t^{-1}}\leq \kappa_{\eta}\eta\sqrt{\lambda_{\min}(M_t)} \}\subseteq \mathcal{B}_{\eta}$.
	\end{lemma}

	Next, in the following lemma, we give an upper bound of $\|Z\|_{M_t^{-1}}=\|G(\hat{\theta}_t)\|_{M_t^{-1}}$, which shows that $G(\hat{\theta}_t)$ and $G(\theta^*)$ are indeed close.
	Its proof is presented in Appendix \ref{section: proof of lemma: G(theta)'s bound 2}.
	\begin{lemma}
		\label{lemma: G(theta)'s bound 2}
		For any $\delta>0$, define the following event:
		\[ \mathcal{E}_G\coloneqq\{\|Z\|_{M_t^{-1}}\leq 4\sqrt{d+\ln (1/\delta)}\}. \]
		Then, $\mathcal{E}_G$ holds with probability at least $1-\delta$.
	\end{lemma}
	
    By the above lemmas, when $\mathcal{E}_G$ holds, for any $\eta$, $\eta\geq\frac{4}{\kappa_{\eta}}\sqrt{\frac{d+\ln(1/\delta)}{\lambda_{\min}(M_t)}}$ implies that $\|\hat{\theta}_t-\theta^*\|\leq\eta$.
    
    It remains to determine appropriate $\eta$ and $\kappa_{\eta}$ for the second step.
    Note that since $\|\hat{\theta}_t-\theta^*\|_1\leq\sqrt{d}\cdot\|\hat{\theta}_t-\theta^*\|\leq\sqrt{d}\eta$, we have $\sum_{e\in[d]}\hat{\theta}(e)\leq\ln(1/\gamma)+\sqrt{d}\eta$.
    Therefore, we choose $\eta=\ln(1+\epsilon)/\sqrt{d}$ for some $\epsilon$ to be determined later.
    In this case, $\kappa_{\eta}=\gamma/(1+\epsilon)\coloneqq\overline{\gamma}$.
    
    To summarize, when $\lambda_{\min}(M_t)\geq \frac{16(d+\ln(1/\delta))}{\kappa_{\eta}^2\eta^2}=\frac{16d(d+\ln(1/\delta))}{\overline{\gamma}^2\ln^2(1+\epsilon)}$, $\|\hat{\theta}_t-\theta^*\|\leq\ln(1+\epsilon)/\sqrt{d}$ with probability at least $1-\delta$.
	
	\paragraph{Step 2: Normality of $\hat{\theta}$.} In the following, we assume that $\hat{\theta}_t$ falls in the $\eta$-neighborhood $\mathcal{B}_\eta$ of $\theta^*$, where $\eta=\ln(1+\epsilon)/\sqrt{d}$, $\kappa_{\eta}=\overline{\gamma}=\gamma/(1+\epsilon)$ and $\epsilon$ is set to be the largest value such that $32^2(1+\epsilon)^6\leq 512(3-\sqrt{2}(1+\epsilon))^2$.
	Define $\Delta\coloneqq\hat{\theta}_t-\theta^*$.
	The previous argument shows that there exists a $s\in[0,1]$ such that
	\[ Z=G(\hat{\theta}_t)-G(\theta^*)=(H+E)\Delta, \]
	where $\bar{\theta}=s\theta^*+(1-s)\hat{\theta}_t\in\mathcal{B}_{\eta}$, $H\coloneqq F(\theta^*)=\sum_{k=1}^{t}\dot{\mu}(X_k^{\top}\theta^*)X_kX_k^{\top}$ and $E\coloneqq F(\bar{\theta})-F(\theta^*)$.
	For any $x\in\mathbb{R}^d$,
	\[ x^{\top}(\hat{\theta}_t-\theta^*)=x^{\top}(H+E)^{-1}Z=x^{\top}H^{-1}Z-x^{\top}H^{-1}E(H+E)^{-1}Z. \]
	Note that $(H+E)^{-1}$ exists since $H+E=F(\bar{\theta})\succ\overline{\gamma} M_t\succ0$.
	We now bound the two terms, respectively.
	
	For the first term, define
	\[ D\coloneqq(X_1,X_2,\cdots,X_t)^{\top}\in\mathbb{R}^{n\times d}. \]
	Note that $D^{\top}D=\sum_{k=1}^{t}X_kX_k^{\top}=M_t$.
	Since $\epsilon_k$ is $1$-sub-Gaussian, by the Hoeffding inequality,
	\[ \Pr[|x^{\top}H^{-1}Z|\geq a]=\Pr\left[\left|\sum_{k=1}^{t}x^{\top}H^{-1}X_k\epsilon_k\right|\geq a\right]\leq \exp\left(-\frac{a^2}{2\|x^{\top}H^{-1}D^{\top}\|^2}\right). \]
	Since $H\succeq \overline{\gamma} M_t$, we have
	\[ \|x^{\top}H^{-1}D^{\top}\|^2=x^{\top}H^{-1}D^{\top}DH^{-1}x\leq\frac{1}{\overline{\gamma}^2}\|x\|_{M_t^{-1}}^2. \]
	Thus we have
	\[ \Pr[|x^{\top}H^{-1}Z|\geq a]\leq \exp\left(-\frac{a^2\overline{\gamma}^2}{2\|x\|_{M_t^{-1}}^2}\right). \]
	By choosing an appropriate $a$, we obtain that with probability at least $1-2\delta$,
	\[ |x^{\top}H^{-1}Z|\leq \frac{\sqrt{2\ln(1/\delta)}}{\overline{\gamma}}\|x\|_{M_t^{-1}}. \]
	
	For the second term,
	\begin{align*}
		|x^{\top}H^{-1}E(H+E)^{-1}Z| &\leq \|x\|_{H^{-1}}\|H^{-1/2}E(H+E)^{-1}Z\| \\
		&\leq \|x\|_{H^{-1}}\|H^{-1/2}E(H+E)^{-1}H^{1/2}\|\|Z\|_{H^{-1}} \\
		&\leq \frac{1}{\overline{\gamma}}\|x\|_{M_t^{-1}}\|H^{-1/2}E(H+E)^{-1}H^{1/2}\|\|Z\|_{M_t^{-1}}.
	\end{align*}
    The first inequality is due to Cauchy-Schwarz inequality.
    The second inequality holds since $\|AB\|\leq\|A\|\|B\|$.
	The last inequality holds since $H\succeq \overline{\gamma} M_t$.
	Next, since $(H+E)^{-1}=H^{-1}-H^{-1}E(H+E)^{-1}$, we have
	\begin{align*}
		\|H^{-1/2}E(H+E)^{-1}H^{1/2}\| &= \|H^{-1/2}E(H^{-1}-H^{-1}E(H+E)^{-1})H^{1/2}\| \\
		&= \|H^{-1/2}EH^{-1/2}-H^{-1/2}EH^{-1}E(H+E)^{-1}H^{1/2}\| \\
		&\leq \|H^{-1/2}EH^{-1/2}\|+\|H^{-1/2}EH^{-1/2}\|\|H^{-1/2}E(H+E)^{-1}H^{1/2}\|
	\end{align*}
    To complete our proof, we need the following technical lemma, whose proof is presented in Appendix \ref{section: proof of lemma: a technical bound}.
    \begin{lemma}
    	\label{lemma: a technical bound}
    	\[ \|H^{-1/2}EH^{-1/2}\|\leq \frac{4}{\overline{\gamma}^2}\sqrt{\frac{d(d+\ln1/\delta)}{\lambda_{\min}(M_t)}}. \]
    	Specifically, when $\lambda_{\min}(M_t)\geq64d(d+\ln(1/\delta))/\overline{\gamma}^4$,
    	\[ \|H^{-1/2}EH^{-1/2}\|\leq 1/2. \]
    \end{lemma}
	Therefore, we have
	\[ \|H^{-1/2}E(H+E)^{-1}H^{1/2}\|\leq\frac{\|H^{-1/2}EH^{-1/2}\|}{1-\|H^{-1/2}EH^{-1/2}\|}\leq 2\|H^{-1/2}EH^{-1/2}\|\leq \frac{8}{\overline{\gamma}^2}\sqrt{\frac{d(d+\ln1/\delta)}{\lambda_{\min}(M_t)}}. \]
	Therefore, together with Lemma \ref{lemma: G(theta)'s bound 2}, we have
	\[ |x^{\top}H^{-1}E(H+E)^{-1}Z|\leq \frac{32\sqrt{d}(d+\ln1/\delta)}{\overline{\gamma}^3\sqrt{\lambda_{\min}(M_t)}}\|x\|_{M_t^{-1}}. \]
	Combining the above inequalities, we have
	\[ |x^{\top}(\hat{\theta}_t-\theta^*)|\leq\left(\frac{\sqrt{2\ln(1/\delta)}}{\overline{\gamma}}+\frac{32\sqrt{d}(d+\ln1/\delta)}{\overline{\gamma}^3\sqrt{\lambda_{\min}(M_t)}}\right)\|x\|_{M_t^{-1}}\leq \frac{3\sqrt{\ln(1/\delta)}}{(1+\epsilon)\overline{\gamma}}\|x\|_{M_t^{-1}}=\frac{3\sqrt{\ln(1/\delta)}}{\gamma}\|x\|_{M_t^{-1}}. \]
	By the choice of $\epsilon$, the last inequality holds when
	\[ \lambda_{\min}(M_t)\geq\frac{512d(d+\ln(1/\delta))^2}{\gamma^4\ln(1/\delta)}. \]
	The proof is completed.
\end{proof}

\subsection{Proof of Lemma \ref{lemma: G(theta)'s bound 1}}
\label{section: proof of lemma: G(theta)'s bound 1}

This lemma is a direct application of Lemma A of \cite{ChenHY99}.
For completeness, we restate it in the lemma below.
\begin{lemma}[\cite{ChenHY99}]
	Let $H$ be a smooth injection from $\mathbb{R}^d$ to $\mathbb{R}^d$ with $H(x_0)=y_0$.
	Define $\mathcal{B}_{\delta}(x_0)\coloneqq\{x\in\mathbb{R}^d\mid \|x-x_0\|\leq\delta\}$ and $\partial\mathcal{B}_{\delta}(x_0)\coloneqq\{x\in\mathbb{R}^d\mid \|x-x_0\|=\delta\}$.
	Then $\inf_{x\in \partial\mathcal{B}_{\delta}(x_0)}\|H(x)-y_0\|\geq r$ implies
	\begin{enumerate}
		\item $\mathcal{B}_r(y_0)\coloneqq\{y\in\mathbb{R}^d\mid \|y-y_0\|\leq r\}\subseteq H(\mathcal{B}_{\delta}(x_0))$.
		\item $H^{-1}(\mathcal{B}_r(y_0))\subseteq \mathcal{B}_{\delta}(x_0)$.
	\end{enumerate}
\end{lemma}
For any $\theta\in \partial\mathcal{B}_{\eta}$, there is some $\bar{\theta}=s\theta+(1-s)\theta^*\in\mathcal{B}_{\eta}$ with $0<s<1$ such that $G(\theta)-G(\theta^*)=F(\bar{\theta})(\theta-\theta^*)$, where $F(\bar{\theta})=\sum_{k=1}^{t}\dot{\mu}(X_k^{\top}\bar{\theta})X_kX_k^{\top}\succeq \kappa_{\eta} M_t$, and
\[ \|G(\theta)\|_{M_t^{-1}}^2=\|G(\theta)-G(\theta^*)\|_{M_t^{-1}}^2=(\theta-\theta^*)^{\top}F(\bar{\theta})M_t^{-1}F(\bar{\theta})(\theta-\theta^*)\geq\kappa_{\eta}^2\lambda_{\min}(M_t)\|\theta-\theta^*\|^2=\kappa_{\eta}^2\eta^2\lambda_{\min}(M_t). \]
By Lemma A of \cite{ChenHY99}, the proof is completed.

\subsection{Proof of Lemma \ref{lemma: G(theta)'s bound 2}}
\label{section: proof of lemma: G(theta)'s bound 2}

Let $\langle\cdot,\cdot\rangle$ denote the inner product. Note that
\[ \|Z\|_{M_t^{-1}}=\|M_t^{-1/2}Z\|=\sup_{\|y\|\leq 1}\langle y, M_t^{-1/2}Z\rangle. \]
Let $\hat{\mathbb{B}}$ be a $1/2$-net of the unit ball $\mathbb{B}^d=\{y\in\mathbb{R}^d\mid \|y\|\leq 1\}$.
Then $|\hat{\mathbb{B}}|\leq 6^d$ \cite{Pollard90}, and for any $x\in\mathbb{B}^d$, there is a $\hat{x}\in\hat{\mathbb{B}}$ such that $\|\hat{x}-x\|\leq1/2$.
Thus,
\begin{align*}
	\langle x, M_t^{-1/2}Z\rangle &=\langle \hat{x}, M_t^{-1/2}Z\rangle+\langle x-\hat{x}, M_t^{-1/2}Z\rangle \\
	&=\langle \hat{x}, M_t^{-1/2}Z\rangle+\|x-\hat{x}\|\langle \frac{x-\hat{x}}{\|x-\hat{x}\|}, M_t^{-1/2}Z\rangle \\
	&\leq \langle \hat{x}, M_t^{-1/2}Z\rangle+\frac{1}{2}\sup_{\|y\|\leq 1} \langle y, M_t^{-1/2}Z\rangle.
\end{align*}
By taking supremum on both sides, we obtain that
\[ \sup_{\|y\|\leq 1} \langle y, M_t^{-1/2}Z\rangle\leq 2\max_{\hat{x}\in\hat{\mathbb{B}}} \langle \hat{x}, M_t^{-1/2}Z\rangle. \]
Finally, define $D\coloneqq (X_1,X_2,\cdots,X_t)^{\top}\in\mathbb{R}^{t\times d}$.
Then, $D^{\top}D=M_t$.
We have
\begin{align*}
	\Pr[\|Z\|_{M_t^{-1}}>a] &\leq \Pr[\max_{\hat{x}\in\hat{\mathbb{B}}} \langle \hat{x}, M_t^{-1/2}Z\rangle> a/2] \\
	&\leq \sum_{\hat{x}\in\hat{\mathbb{B}}}\Pr[\langle \hat{x}, M_t^{-1/2}Z\rangle> a/2] \\
	&\leq \sum_{\hat{x}\in\hat{\mathbb{B}}}\exp\left(-\frac{a^2}{8\|\hat{x}^{\top}M_t^{-1/2}D^{\top}\|^2}\right) \\
	&\leq \exp\left(-a^2/8+d\ln 6\right)\leq\delta.
\end{align*}
The second to last inequality holds due to Hoeffding inequality.
The last inequality holds by choosing $a=4\sqrt{d+\ln(1/\delta)}$.

\subsection{Proof of Lemma \ref{lemma: a technical bound}}
\label{section: proof of lemma: a technical bound}

By the mean value theorem,
\[ E=\sum_{k=1}^{t}(\dot{\mu}(X_k^{\top}\bar{\theta})-\dot{\mu}(X_k^{\top}\theta^*))X_kX_k^{\top}=\sum_{k=1}^{t}\ddot{\mu}(r_k)X_k^{\top}\Delta X_kX_k^{\top}. \]
for some $r_k\in\mathbb{R}$. Since $|\ddot{\mu}|\leq 1$, for any $x\in\mathbb{R}^d\setminus\{0\}$, we have
\begin{align*}
	x^{\top}H^{-1/2}EH^{-1/2}x &=\sum_{k=1}^{t}\ddot{\mu}(r_k)X^{\top}_k\Delta\|x^{\top}H^{-1/2}X_k\|^2 \\
	&\leq \sqrt{d} \|\Delta\|\left(x^{\top}H^{-1/2}\left(\sum_{k=1}^{t}X_kX_k^{\top}\right)H^{-1/2}x\right) \\
	&\leq \frac{\sqrt{d}}{\overline{\gamma}}\|\Delta\|\|x\|^2.
\end{align*}
The first inequality is due to Cauchy-Schwarz inequality.
The second inequality holds since $\|X_k\|\leq\sqrt{d}$.
The last inequality holds since $H\succ\overline{\gamma} M_t$.
Therefore, by the definition of spectral norm of a matrix,
\[ \|H^{-1/2}EH^{-1/2}\|\leq \frac{\sqrt{d}}{\overline{\gamma}}\|\Delta\|\leq \frac{4}{\overline{\gamma}^2}\sqrt{\frac{d(d+\ln1/\delta)}{\lambda_{\min}(M_t)}}. \]
When $\lambda_{\min}(M_t)\geq64d(d+\ln(1/\delta))/\overline{\gamma}^4$, we have
\[ \|H^{-1/2}EH^{-1/2}\|\leq 1/2. \]

\section{Proof of Lemma \ref{lemma: TPM}}
\label{section: proof of lemma: TPM}

In this section, we present the proof of Lemma \ref{lemma: TPM}, the GOM bounded smoothness for the IC model.
The proof is similar to that for the LT model in \citep{LiKTLC20}.

Fix any seed set $S\subseteq V$ throughout the proof.
For edge probability vector $p^*$, let $\sigma(S,p^*,v)$ be the probability that $v$ is activated under $p^*$.
Let $G'$ be a random sub-graph of the original graph $G$ such that each edge $e\in E$ of $G$ appears in $G'$ independently with probability $p^*(e)$.
By the diffusion rule of the IC model, $\sigma(S,p^*,v)$ also means the probability that there is a path from $S$ to $v$ in $G'$.
For two distinct edge probability vectors $p^*$ and $\tilde{p}$, we couple them in the following way.
Define a random vector $r\sim U[0,1]^m$ which satisfies that $r(e)\sim U[0,1]$ for each edge $e\in E$, where $U[0,1]$ denotes the uniform distribution over interval $[0,1]$.
Given edge probability vector $p^*$, for $e\in E$, $e$ appears in $G'$ if and only if $r(e)\leq p^*(e)$.
It is easy to see that the probability that $e$ appears in $G'$ is $p^*(e)$.
For $\tilde{p}$, the random experiment can be executed using the same $r$.
In this way, $p^*$ and $\tilde{p}$ are coupled.

For node $v\in V$, define event
\[ \mathcal{E}_{0,v}\coloneqq \mathbf{1}\{v \,\mbox{is activated under}\,\tilde{p}\}\neq \mathbf{1}\{v \,\mbox{is activated under}\,p^*\}. \]
According to our notations,
\[ |\sigma(S,\tilde{p})-\sigma(S,p^*)| =\left|\sum_{v\in V\setminus S} \sigma(S,\tilde{p},v)-\sigma(S,p^*,v)\right|\leq \sum_{v\in V\setminus S}|\sigma(S,\tilde{p},v)-\sigma(S,p^*,v)|=\sum_{v\in V\setminus S}\Pr_{r\sim U[0,1]^m}[\mathcal{E}_{0,v}]. \]
Recall that $V[S,v]$ is defined to be the set of nodes relevant to $v$ given seed set $S$, namely for each $u\in V[S,v]$, there is a path $P$ from $S$ to $v$ such that $u\in P$.
Thus, the occurrence of $\mathcal{E}_{0,v}$ means there is a node $u\in V[S,v]$ such that its activation state is different under $\tilde{p}$ and $p^*$.
More strictly, let $\Phi(p^*,r)\coloneqq (S_0=S,S_1,\ldots,S_{n-1})$ be the sequence of active nodes under $p^*$ and $r$.
Let $\Phi_{\tau}(p^*,r)=S_{\tau}$ be the set of active nodes immediately after time $\tau$.
For node $u\in V[S,v]$, define $\mathcal{E}_1(u)$ to describe that $u$ is the first node that possesses distinct activation states under $\tilde{p}$ and $p^*$, namely
\begin{align*}
	\mathcal{E}_1(u)\coloneqq\{r\mid \exists\,\tau,\forall\,\tau'<\tau, \Phi_{\tau'}(p^*,r)=\Phi_{\tau'}(\tilde{p},r), u\in (\Phi_{\tau}(p^*,r)\setminus \Phi_{\tau}(\tilde{p},r))\cup(\Phi_{\tau}(\tilde{p},r)\setminus \Phi_{\tau}(p^*,r))\}.
\end{align*}
Then, by the above argument,
\[ \mathcal{E}_{0,v}\subseteq\bigcup_{u\in V[S,v]} \mathcal{E}_1(u). \]
Therefore, by union bound,
\begin{equation}
	\label{eq: aux3}
	|\sigma(S,\tilde{p})-\sigma(S,p^*)|\leq \sum_{v\in V\setminus S}\sum_{u\in V[S,v]} \Pr_{r\sim U[0,1]^m}[\mathcal{E}_1(u)].
\end{equation}

To further determine the probability that $\mathcal{E}_1(u)$ occurs, for each $0\leq \tau\leq n-1$, we define the following event:
\begin{align*}
	\mathcal{E}_{2,0}(u,\tau) &\coloneqq\{r\mid \forall\,\tau'<\tau, \Phi_{\tau'}(p^*,r)=\Phi_{\tau'}(\tilde{p},r),u\notin \Phi_{\tau-1}(p^*,r)\}. \\
	\mathcal{E}_{2,1}(u,\tau) &\coloneqq\{r\mid \forall\,\tau'<\tau, \Phi_{\tau'}(p^*,r)=\Phi_{\tau'}(\tilde{p},r), u\in (\Phi_{\tau}(p^*,r)\setminus \Phi_{\tau}(\tilde{p},r))\cup(\Phi_{\tau}(\tilde{p},r)\setminus \Phi_{\tau}(p^*,r))\}. \\
	\mathcal{E}_{3,1}(u,\tau) &\coloneqq\{r\mid u\in (\Phi_{\tau}(p^*,r)\setminus \Phi_{\tau}(\tilde{p},r))\cup(\Phi_{\tau}(\tilde{p},r)\setminus \Phi_{\tau}(p^*,r))\}.
\end{align*}
Note that event $\mathcal{E}_{2,1}(u,\tau)$ means that event $\mathcal{E}_1(u)$ occurs in time $\tau$.
Hence,
\[ \Pr_{r\sim U[0,1]^m}[\mathcal{E}_1(u)]=\sum_{\tau=0}^{n-1}\Pr_{r\sim U[0,1]^m}[\mathcal{E}_{2,1}(u,\tau)]. \]

Next, we estimate the value of $\Pr_{r\sim U[0,1]^m}[\mathcal{E}_{2,1}(u,\tau)]$.
For $u\in V$, let $r_u\in[0,1]^{d_u}$ denote the sub-vector which consists of $r$'s entries over the incoming edges of $u$.
Let $r_{-u}\in[0,1]^{m-d_u}$ denote the vector consisting of the remaining entries.
By fixing $r_{-u}$, we define sub-event $\mathcal{E}_{2,1}(u,\tau,r_{-u})\subseteq \mathcal{E}_{2,1}(u,\tau)$ to be the restriction of event $\mathcal{E}_{2,1}(u,\tau)$ when $r_{-u}$ is fixed.
Similarly, we can define $\mathcal{E}_{2,0}(u,\tau,r_{-u})\subseteq \mathcal{E}_{2,0}(u,\tau)$ and $\mathcal{E}_{3,1}(u,\tau,r_{-u})\subseteq \mathcal{E}_{3,1}(u,\tau)$.
Note that
\begin{align*}
	\mathcal{E}_{2,1}(u,\tau) &=\mathcal{E}_{2,0}(u,\tau)\cap \mathcal{E}_{3,1}(u,\tau). \\
	\mathcal{E}_{2,1}(u,\tau,r_{-u}) &=\mathcal{E}_{2,0}(u,\tau)\cap \mathcal{E}_{3,1}(u,\tau,r_{-u}).
\end{align*}
Therefore,
\[ \Pr_{r\sim U[0,1]^m}[\mathcal{E}_{2,1}(u,\tau)] =\Pr_{r\sim U[0,1]^m}[\mathcal{E}_{2,0}(u,\tau)]\cdot\Pr_{r\sim U[0,1]^m}[\mathcal{E}_{3,1}(u,\tau)\mid \mathcal{E}_{2,0}(u,\tau)] \leq \Pr_{r\sim U[0,1]^m}[\mathcal{E}_{3,1}(u,\tau)\mid \mathcal{E}_{2,0}(u,\tau)], \]
and
\begin{equation}
	\label{eq: aux1}
	\begin{split}
		\Pr_{r_u\sim U[0,1]^{d_u}}[\mathcal{E}_{2,1}(u,\tau,r_{-u})] &=\Pr_{r_u\sim U[0,1]^{d_u}}[\mathcal{E}_{2,0}(u,\tau,r_{-u})]\cdot\Pr_{r_u\sim U[0,1]^{d_u}}[\mathcal{E}_{3,1}(u,\tau,r_{-u})\mid \mathcal{E}_{2,0}(u,\tau,r_{-u})] \\
		&\leq \Pr_{r_u\sim U[0,1]^{d_u}}[\mathcal{E}_{3,1}(u,\tau,r_{-u})\mid \mathcal{E}_{2,0}(u,\tau,r_{-u})].
	\end{split}
\end{equation}
By the definition of event $\mathcal{E}_{2,0}(u,\tau,r_{-u})$, by the end of time $\tau-1$, the sets of active nodes are the same under $p^*$ and $\tilde{p}$, and $u$ is inactive at the time.
Besides, since $r_{-u}$ is fixed, the set of active nodes by the end of time $\tau-1$ is also fixed.
We use $\Phi_{\tau'}(\mathcal{E}_{2,0}(u,\tau,r_{-u}))$ to denote the set of active nodes by the end of time $\tau'$ under event $\mathcal{E}_{2,0}(u,\tau,r_{-u})$.
Consider the probability that event $\mathcal{E}_{3,1}(u,\tau,r_{-u})$ occurs conditioned on event $\mathcal{E}_{2,0}(u,\tau,r_{-u})$.
Define $Q(u,\tau,r_{-u})$ to be the neighbors of $u$ which were just activated in time $\tau-1$, namely
\[ Q(u,\tau,r_{-u})\coloneqq (\Phi_{\tau-1}(\mathcal{E}_{2,0}(u,\tau,r_{-u}))\setminus \Phi_{\tau-2}(\mathcal{E}_{2,0}(u,\tau,r_{-u})))\cap N(u). \]
Let $Z=Z(u,\tau,r_{-u})\in\{0,1\}^{d_u}$ be the characteristic vector of $Q(u,\tau,r_{-u})$.
By the diffusion rule of the IC model,
\[ \Pr_{r_u\sim U[0,1]^{d_u}}[u\in \Phi_{\tau}(p^*,r) \mid \mathcal{E}_{2,0}(u,\tau,r_{-u})]=1-\prod_{u'\in Q} (1-p^*(e_{u'u}))=\mu(Z^{\top}\theta^*_u), \]
where the definition of $\theta^*_u$ can be found in eq.~\eqref{eq: para trans} and $\mu:\mathbb{R}\rightarrow\mathbb{R}$ is the link function which has the form $\mu(x)=1-\exp(-x)$.
Therefore,
\[ \Pr_{r_u\sim U[0,1]^{d_u}}[\mathcal{E}_{3,1}(u,\tau,r_{-u}) \mid \mathcal{E}_{2,0}(u,\tau,r_{-u})]\leq |\mu(Z^{\top}\theta^*_u)-\mu(Z^{\top}\tilde{\theta}_u)|\leq |Z^{\top}(\theta^*_u-\tilde{\theta}_u)|. \]
The last inequality holds since $\mu$ is $1$-Lipschitz.
Combining with eq.~\eqref{eq: aux1}, we obtain that
\begin{equation}
	\label{eq: aux2}
	\Pr_{r_u\sim U[0,1]^{d_u}}[\mathcal{E}_{2,1}(u,\tau,r_{-u})]\leq |Z(u,\tau,r_{-u})^{\top}(\theta^*_u-\tilde{\theta}_u)|.
\end{equation}

Event $\mathcal{E}_{2,0}(u,\tau,r_{-u})$ depends on both $p^*$ and $\tilde{p}$, which is hard to analyze.
For this reason, we define event
\[ \mathcal{E}_{4,0}(u,\tau,r_{-u})\coloneqq \{\theta=(r_{-u},r_u)\mid u\notin\Phi_{\tau-1}(p^*,r)\}. \]
Clearly, $\mathcal{E}_{2,0}(u,\tau,r_{-u})\subseteq \mathcal{E}_{4,0}(u,\tau,r_{-u})$.
Besides, when $\mathcal{E}_{2,0}(u,\tau,r_{-u})\neq \emptyset$, since $r_{-u}$ is fixed in events $\mathcal{E}_{2,0}(u,\tau,r_{-u})$ and $\mathcal{E}_{4,0}(u,\tau,r_{-u})$ and $u$ is inactive by the end of time $\tau-1$, the active nodes by the end of time $\tau-1$ are the same under the two events.
Define $P(u,\tau,r_{-u})$ be the set of $u$'s neighbors which are just activated in time $\tau-1$ under event $\mathcal{E}_{4,0}(u,\tau,r_{-u})$, namely
\[ P(u,\tau,r_{-u})\coloneqq (\Phi_{\tau-1}(\mathcal{E}_{4,0}(u,\tau,r_{-u}))\setminus \Phi_{\tau-2}(\mathcal{E}_{4,0}(u,\tau,r_{-u})))\cap N(u). \]
Then, $Q(u,\tau,r_{-u})=P(u,\tau,r_{-u})$.
Let $X(u,\tau,r_{-u})\in\{0,1\}^{d_u}$ be the characteristic vector of $P(u,\tau,r_{-u})$.
By the above argument and eq.~\eqref{eq: aux2},
\[ \Pr_{r_u\sim U[0,1]^{d_u}}[\mathcal{E}_{2,1}(u,\tau,r_{-u})]\leq |X(u,\tau,r_{-u})^{\top}(\theta^*_u-\tilde{\theta}_u)|. \]
When $\mathcal{E}_{2,0}(u,\tau,r_{-u})= \emptyset$, the left-hand side of the above inequality is $0$.
Thus, the inequality still holds.
To sum up,
\begin{align*}
	\Pr_{r\in U[0,1]^m}[\mathcal{E}_1(u)]
			&=\int_{r_{-u}\sim [0,1]^{m-d_u}} \sum_{\tau=0}^{n-1}\Pr_{r_u\sim U[0,1]^{d_u}}[\mathcal{E}_{2,1}(u,\tau,r_{-u})]\,\mathrm{d}r_{-u} \\
			&\leq \int_{r_{-u}\in [0,1]^{m-d_u}} \sum_{\tau=0}^{n-1}|X(u,\tau,r_{-u})^{\top}(\theta^*_u-\tilde{\theta}_u)|\,\mathrm{d}r_{-u} \\
			&=\E_{r_{-u}\in U[0,1]^{m-d_u}}\left[\sum_{\tau=0}^{n-1}|X(u,\tau,r_{-u})^{\top}(\theta^*_u-\tilde{\theta}_u)|\right].
\end{align*}
By plugging the above inequality into eq.~\eqref{eq: aux3}, we obtain that
\begin{align*}
	|\sigma(S,\tilde{p})-\sigma(S,p^*)|
			&\leq \sum_{v\in V\setminus S}\sum_{u\in V[S,v]} \E_{r_{-u}\in U[0,1]^{m-d_u}}\left[\sum_{\tau=0}^{n-1}|X(u,\tau,r_{-u})^{\top}(\theta^*_u-\tilde{\theta}_u)|\right] \\
			&\leq \sum_{v\in V\setminus S}\sum_{u\in V[S,v]}\E\left[\sum_{j=1}^{J_{u}}\left|X_{j,u}^{\top}(\tilde{\theta}_u-\theta^*_u)\right|\right].
\end{align*}
The last inequality holds since by definition, $X_{j,u}\in\{0,1\}^{d_u}$ indicates the characteristic vector of the edges corresponding to the $j$-th batch of neighbors that were activated.

\section{Proof of Lemma \ref{lemma: upper bound of modified M norm}}
\label{section: proof of lemma: upper bound of modified M norm}

Fix $v\in V$.
For simplicity, let $z_{t,j}=\|X_{t,j,v}\|_{M_{t-1,v}^{-1}}$
Recall the definition of $M_{t,v}$ (see eq.~\eqref{eq: data matrix}), we have
\[ M_{t,v} = M_{t-1,v}+\sum_{j=1}^{J_{t,v}}X_{t,j,v}X_{t,j,v}^{\top}. \]
Basic algebra gives us that
\begin{align*}
	\det[M_{t,v}]
	&\geq\det\left[M_{t-1,v}+X_{t,j,v}X_{t,j,v}^{\top}\right] \\
	&=\det\left[M_{t-1,v}^{1/2}\left(I+M_{t-1,v}^{-1/2}X_{t,j,v}X_{t,j,v}^{\top}M_{t-1,v}^{-1/2}\right)M_{t-1,v}^{1/2}\right] \\
	&=\det[M_{t-1,v}]\det\left[I+M_{t-1,v}^{-1/2}X_{t,j,v}X_{t,j,v}^{\top}M_{t-1,v}^{-1/2}\right] \\
	&=\det[M_{t-1,v}]\left(1+X_{t,j,v}^{\top}M_{t-1,v}^{-1}X_{t,j,v}\right) \\
	&=\det[M_{t-1,v}](1+z_{t,j}^2).
\end{align*}
Thus, it can be deduced that
\[ \det[M_{t,v}]^{J_{t,v}} \geq \det[M_{t-1,v}]^{J_{t,v}}\prod_{j=1}^{J_{t,v}}(1+z_{t,j}^2). \]
Next, by $\det[M_{t,v}]\geq \det[M_{t-1,v}]$ and $J_{t,v}\leq d_v$, we have
\[ \det[M_{t,v}]^{d_v} \geq \det[M_{t-1,v}]^{d_v}\prod_{j=1}^{J_{t,v}}(1+z_{t,j}^2). \]
Therefore, we have
\begin{align*}
	\det[M_{T,v}]^{d_v}\geq\det[M_{T_0,v}]^{d_v}\prod_{t=T_0+1}^{T}\prod_{j=1}^{J_{t,v}}(1+z_{t,j}^2)=R^{d_v^2}\prod_{t=T_0+1}^{T}\prod_{j=1}^{J_{t,v}}(1+z_{t,j}^2).
\end{align*}
On the other hand,
\begin{align*}
	\tr[M_{T,v}] =\tr\left[M_{T_0,v}+\sum_{t=T_0+1}^{T}\sum_{j=1}^{J_{t,v}}X_{t,j,v}X_{t,j,v}^{\top}\right]=\tr[M_{T_0,v}]+\sum_{t=T_0+1}^{T}\sum_{j=1}^{J_{t,v}}\|X_{t,j,v}\|^2\leq Rd_v+(T-T_0)d_v^2.
\end{align*}
By the trace-determinant inequality, we have
\begin{align*}
	(R+(T-T_0)d_v)^{d_v} \geq\left(\frac{1}{d_v}\tr[M_{T,v}]\right)^{d_v} \geq\det[M_{T,v}] \geq R^{d_v^2}\prod_{t=T_0+1}^{T}\prod_{j=1}^{J_{t,v}}(1+z_{t,j}^2).
\end{align*}
By taking logarithm on both sides, we obtain that
\begin{align*}
	d_v\ln(R+(T-T_0)d_v)\geq \sum_{t=T_0+1}^T\sum_{j=1}^{J_{t,v}}\ln(1+z_{t,j}^2)+d_v^2\ln R\geq \sum_{t=T_0+1}^T\sum_{j=1}^{J_{t,v}}\frac{z_{t,j}^2}{d_v+1}+d_v\ln R^{d_v}.
\end{align*}
The last inequality holds since $\ln(1+x)\geq \frac{x}{d_v+1}$ for $x\in[0,d_v]$, and
\[ z_{t,j}^2=\|X_{t,j,v}\|_{M_{t-1,v}^{-1}}^2\leq\|X_{t,j,v}\|^2\leq d_v. \]
By rearranging the above inequality, we have
\begin{align*}
	\sum_{t=T_0+1}^T\sum_{j=1}^{J_{t,v}}\|X_{t,j,v}\|_{M_{t-1,v}^{-1}}^2 =\sum_{t=T_0+1}^T\sum_{j=1}^{J_{t,v}}z_{t,j}^2 \leq d_v(d_v+1)\cdot\ln\left(\frac{R+(T-T_0)d_v}{R^{d_v}}\right).
\end{align*}
Finally, by the Cauchy-Schwarz inequality,
\begin{align*}
	\sum_{t=T_0+1}^{T}\sum_{v\in V\setminus S_t}n_{S_t,v}\sum_{j=1}^{J_{t,v}}\|X_{t,j,v}\|_{M_{t-1,v}^{-1}}
	&\leq \sqrt{\left(\sum_{t=T_0+1}^{T}\sum_{v\in V\setminus S_t}\sum_{j=1}^{J_{t,v}}n_{S_t,v}^2\right)\left( \sum_{t=T_0+1}^{T}\sum_{v\in V\setminus S_t}\sum_{j=1}^{J_{t,v}}\|X_{t,j,v}\|_{M_{t-1,v}^{-1}}^2\right)} \\
	&\leq \sqrt{\left( \sum_{t=T_0+1}^{T}\sum_{v\in V}d_v n_{S_t,v}^2\right)\left( \sum_{v\in V}d_v(d_v+1)\cdot\ln\left(\frac{R+(T-T_0)d_v}{R^{d_v}}\right)\right)} \\
	&\leq \zeta(G)D\sqrt{(m+n)(T-T_0)\ln\left(R+(T-T_0)D\right)}.
\end{align*}
The last inequality holds since by the definition,
\[ \sqrt{\sum_{v\in V}n_{S_t,v}^2}\leq\zeta(G). \]

\end{document}